\documentclass[runningheads,orivec]{llncs}
\usepackage[T1]{fontenc}

\usepackage{graphicx}
\usepackage[normalem]{ulem}
\usepackage{amsmath}
\usepackage{amssymb}
\usepackage{mdframed}
\usepackage{hyperref}
\usepackage{tabularx}
\usepackage{nicefrac}
\usepackage{marginnote}
\usepackage{marginnote}
\usepackage{xcolor}
\usepackage{thm-restate}
\usepackage[title]{appendix}
\usepackage[most]{tcolorbox}

\usepackage{color}

\usepackage{listings}

\newcommand{\CodeSymbol}[1]{\textcolor{orange}{#1}}
\lstdefinelanguage{ocaml}{
  literate={=}{{\CodeSymbol{=}}}1
           {->}{{\CodeSymbol{->}}}1
           {<-}{{\CodeSymbol{<-}}}1,
  keywords={type, mutable, and, list, of, bool, option, ref, let, then, if, function, string},
  keywordstyle=\color{blue}\bfseries,
  keywords=[2]{Type},
  keywordstyle=[2]\color{blue}\bfseries,
  identifierstyle=\color{black},
  sensitive=false,
  mathescape=true,
  comment=[l]{//},
  morecomment=[n]{(*}{*)},
  commentstyle=\color{purple}\ttfamily,
  stringstyle=\color{red}\ttfamily,
  morestring=[b]',
  morestring=[b]"
}

\newcommand{\refitem}[2]{\autoref{#1}.\ref{#2}}

\makeatletter
\providecommand{\leftsquigarrow}{%
  \mathrel{\mathpalette\reflect@squig\relax}%
}
\newcommand{\reflect@squig}[2]{%
  \reflectbox{$\m@th#1\rightsquigarrow$}%
}
\makeatother

\newcommand{\lrtrans}{%
  \mathrel{%
    \vcenter{\offinterlineskip
      \ialign{##\cr$\leftsquigarrow$\cr\noalign{\kern0.5pt}$\rightsquigarrow$\cr}%
    }%
  }%
}

\newcolumntype{C}{>{\centering\arraybackslash}X}

\newcounter{notecounter}[section]

\newcounter{notecounternp}[section]

\newcommand{\tgets}{\mkern-4mu\gets\mkern-4mu}
\newcommand{\rback}[1]{{#1}_\downarrow}
\newcommand{\crback}[1]{\underline{{#1}}_\downarrow}

\newcommand{\mone}{\to_{m_1}}
\newcommand{\mtwo}{\to_{m_2}}
\newcommand{\crumb}[1]{\underline{{#1}}}
\newcommand{\es}[2]{\left[ {#1} \tgets {#2} \right]}
\newcommand{\ess}[1]{\left[ {*} \tgets {#1} \right]}
\newcommand{\sub}[2]{\{ {#1} \tgets {#2} \}}
\newcommand{\plug}[1]{\langle {#1} \rangle}
\newcommand{\size}[1]{\left|{#1}\right|}
\newcommand{\len}[1]{\left|{#1}\right|_{len}}
\def\hole{\plug{\cdot}}
\DeclareMathOperator{\fv}{fv}
\DeclareMathOperator{\bv}{bv}
\DeclareMathOperator{\dom}{dom}

\def\bvr{$\to_{\beta_v}$}
\newcommand\crumbling{finest crumbling}

\newcommand{\mmoner}{\rightsquigarrow_{m_1}}
\newcommand{\mmonel}{ \leadsto_{m^b_1} }
\newcommand{\mmonelr}{ \overset{m_1^b}{\underset{m_1}{\lrtrans}} }

\newcommand{\mmtwolr}{ \overset{m_2^b}{\underset{m_2}{\lrtrans}} }
\newcommand{\mmtwor}{\rightsquigarrow_{m_2}}
\newcommand{\mmtwol}{\rightsquigarrow_{m_2^b}}

\newcommand{\sealr}{ \overset{sea^b}{\underset{sea}{\lrtrans}} }
\newcommand{\sear}{\rightsquigarrow_{sea}}
\newcommand{\seal}{\rightsquigarrow_{sea^b}}

\newcommand\rcr{\leadsto_{f}}

\newcommand\rcrinv{\leadsto_{b}}

\newcommand\rcrfull{\leadsto_{rCr}}

\newcommand{\rbackm}[1]{{#1}_\Downarrow}
\newcommand{\initm}[1]{\iota\left({#1}\right)}
\newcommand\hist{\mathcal{H}}

\def\hsep{\vspace{1ex} \hrule \vspace{1ex}}

\begin{document}

\title{A Reversible Crumbling Abstract Machine for Plotkin's Call-by-Value\thanks{Partially funded by the INdAM/GNCS project MARQ (CUP\_E53C24001950001)}}

\author{
	Nicol\`o Pizzo\inst{1}\orcidID{0009-0003-4658-0178} \and
	Claudio Sacerdoti Coen\inst{1}\orcidID{0000-0002-4360-6016}
}

\institute{Alma Mater Studiorum - Università di Bologna}

\maketitle     

\begin{abstract}
	Landauer's embeddings enable the reversibility of computations for non-reversible programming languages, augmenting each intermediate state with enough data to reconstruct the previous state. An interesting research question is therefore to try to reduce the space overhead required.
In this work we propose a Landauer's embedding for Plotkin's call-by-value calculus (CbV). In order to control the computational complexity of CbV and turn the number of $\beta$-steps into a cost model, CbV is typically implemented via reduction machines.
We show that one machine, that has not received much attention, exhibits a particularly compact Landauer's embedding, requiring only constant space overhead for each step.


\end{abstract}

\section{Introduction}
The goal of this paper is to equip Plotkin's weak Call-by-Value \cite{plotkinCallbynameCallbyvalueLcalculus1975} calculus with a Landauer's embedding~\cite{landauerIrreversibilityHeatGeneration1961b}. Doing so in a naive way would result in an expensive cost, at each step linear in the size of the term to be reduced, as for each reduction step we need to remember the position of the redex and enough information to recover it from the reduct. For example, to recover $(\lambda x.y)M$ from its reduct $y$ one needs to remember all of $M$.
Moreover, it is well-known~\cite{accattoliExponentialsSubstitutionsCost2023} that under weak CbV there are families of $\lambda$-terms that incur in the phenomenon of size-explosion, where intermediate terms may grow exponentially in the number of steps with respect to the initial term. Therefore nobody implements the $\lambda$-calculus naively. Consider for example the term $(\lambda x. x x \ldots x) v$ that reduces to $v v \ldots v $: the cost of this step entails the cost of copying $v$ as many times as the number of $x$, and when the size of the argument $v$ is roughly equal to the size of the abstraction, the total cost is quadratic in the size of the term (and the Landauer's embedding would require a linear amount of space to recover the redex).

All efficient implementations of $\lambda$-calculi avoid the size explosion by replacing immediate substitutions with some form of
explicit or delayed substitution, eventually collected in larger data structures called environments. With explicit substitutions
$(\lambda x. x x \ldots x) v$ reduces in constant time to $(x x \ldots x)\es{x}{v}$ where $\es{x}{v}$ is sometimes written as
\verb|let x = v in t|. It is also obvious that a Landauer's embedding now only needs to record the position of the reduct to fully restore
the redex. However, remembering the position of a redex still needs a linear amount of information on the size of the term and one
can wonder if it is possible to use only a constant size. Moreover some explicitly substituted occurrences need to trigger the substitution
anyway to advance in the computation and thus more ingredients are required.

\paragraph{Abstract Machines.}
Abstract reduction machines provide for efficient implementations of $\lambda$-calculi where all details are taken care. In particular an
abstract machine uses explicit substitutions (ES), controls their firing and implements steps to search for the next redex. The most efficient
machines for CbV have a complexity that is only bi-linear in the number of reduction steps of the calculus and in the size of the initial
term to be evaluated. Among them there is a family of machines that are called Crumbling Abstract Machines (CAM)~\cite{accattoliCrumblingAbstractMachines2019}.

In this work we show a new variant of CAM whose complexity is indeed bi-linear and that admits a particularly compact Landauer's embedding,
requiring just additional O(1) space for each machine step. We call the resulting machine the Reversible Crumbling Abstract Machine (RCAM).

\paragraph{Crumbling.}
Equipping $\lambda$-calculus with sharing enables a representation of terms where a sharing point can be placed wherever the constructor of a term occurs. This is the same principle used for A-normal forms~\cite{flanaganEssenceCompilingContinuations2004}. For example, the term $(yx)(xy)$ can be represented as $(z_1z_2)\es{z_1}{yx}\es{z_2}{xy}$.
The crumbling technique \cite{accattoliCrumblingAbstractMachines2019} applies this concept hereditarily to completely avoid the immediate nesting of applications: for instance $(x_1(x_2\lambda x_3.x_3))x_4$ is written as $(y_1x_4)\es{y_1}{x_1y_2}\es{y_2}{x_2\lambda x_3.x_3}$. The various variants of crumbling differ in what is still allowed in the terms to be substituted. For example, in some variants $\es{y_2}{x_2\lambda x_3.x_3}$ is not allowed and must become $\es{y_2}{x_2y_3}\es{y_3}{\lambda x_3.x_3}$. Not every variant has been exploited in the literature so far: they are introduced and studied when interesting properties of them are noticed, like in~\cite{accattoliPositiveSharingAbstract2025} where one variant is noticed to be related to Miller's positive $\lambda$-calculus \cite{millerPositivePerspectiveTerm2023}.

This work builds upon the concept of crumbling abstract machines presented in \cite{accattoliCrumblingAbstractMachines2019}, and aims at showing a new variant of the crumbling for which the Landauer's embedding becomes trivial and economic. We dub the variant the \emph{\crumbling\ }because it is the one that inserts the highest number of explicit substitutions by only allowing as terms to be substituted either
$\lambda$-abstractions or applications of a variable to another variable. In particular, by ruling out also the case of a variable substituted by another variable, it fully embraces the notion of practical variables~\cite{accattoliValueVariables2014} by ruling out variables from the
set of values where they semantically belong, leading to speed ups in implementations of reduction machines.

This work shows that the \crumbling\ is also an interesting variant, as it enables to achieve a Landauer's embedding with a reasonable spatial overhead.

\paragraph{Contributions.}
The contributions of this work consist in presenting an economic Landauer's embedding for Plotkin's closed CbV by means of a Landauer's
embedding for a new reduction machine for the \crumbling, that we call \emph{Reversible Crumbling Abstract Machine (RCAM)}. In addition to
showing the reversibility of the machine, we establish the usual properties of crumbling machines, i.e. that the machine implements
correctly the calculus and that the computational complexity of a computation made of forward steps only is bi-linear in the number of
reduction steps of the calculus and the size of the initial term to be reduced.

\definecolor{green}{RGB}{0, 153, 51}

\section{The calculus} \label{sec:calculus}
In this section we first introduce Plotkin's (closed) CbV calculus in its right-to-left variant, describing the basic notions employed by the calculus. We then move on to present the \crumbling\ -- a variant of the crumbling technique -- that is part of the focus of this work. 

\subsection{Plotkin's call-by-value calculus} \label{sec:plotkin-calculus}
\begin{figure}[!t]
	\begin{mdframed}
		\centering
		\begin{tabularx}{\textwidth}{CC}
			\textsc{Terms}
      & \textsc{Values} \\

			$u, t ::= v \mid ut$
      & $v := x \mid \lambda x. t$ 
		\end{tabularx}
    \begin{tabularx}{\textwidth}{CC}
      \textsc{Right v-context}
      & \textsc{Reduction Rule} \\
			
      $R ::= \hole \mid tR \mid Rv$
      & $R \plug{(\lambda x. t) v} \to_{\beta_v} R\plug{t \sub{x}{v}}$
    \end{tabularx}
	\end{mdframed}

	\caption{Plotkin's Closed CbV calculus}
	\label{fig:plot-calculus}
\end{figure}

In \emph{Plotkin's CbV ($\lambda_{Plot}$)} terms are defined mutually with values: a term is either a value, or an application of terms; values are either variables or $\lambda$-abstractions. In this work we specifically address Plotkin's closed (weak) CbV.
In this paper we mainly use the notion of (weak) \textit{right v-context} (see \autoref{fig:plot-calculus}) where holes are outside abstractions and on the right of the hole may only appear values. 

\paragraph{Evaluating Terms.}
Terms are evaluated in weak (i.e. outside abstractions) Call-by-Value: a $\to_{\beta_v}$ reduction step is triggered by plugging in a context the application of a value to a $\lambda$-abstraction, called a \emph{redex}.
The reduction step then results in plugging the body of the abstraction where we {\em immediately substitute} the bound variable with the applied value. The particular shape of right v-contexts forces the rightmost redex to be the first evaluated. To change the evaluation order
one would just pick left v-contexts instead and, mutatis mutandis, all the rest of the paper would stand.

An important property of Plotkin's closed CbV is \textit{harmony} that at once characterizes normal forms and rules out
stuck computations.
\begin{restatable}[Plotkin's CbV harmony]{lemma}{plotkinharmony}
	Let $t$ be a closed term. Then $t$ it is $\beta_v$-normal if and only if $t$ is a value.
\end{restatable}

\subsection{The finest crumbled calculus} \label{sec:crumbled-calculus}
\begin{figure}[!t]
	\begin{mdframed}
		\centering
		\begin{tabularx}{\linewidth}{CC}
			\textsc{Bites}
      & \textsc{Crumbled Values} \\
			$b ::= xy \mid v$ 
      & $v ::= \lambda x. \ess{y} \mid \lambda x. E$ \\ [1ex]

      \textsc{Environments}
      & \textsc{Environment Context} \\
      $E ::= \epsilon \mid E\es{x}{b}$
      & $C ::= E\es{x}{\hole}$
		\end{tabularx}

		\hsep
		\textsc{Reduction Rules}
		\begin{align*}
			E \es{z}{xy} E_v & \mone E \es{z}{b'}E'E_v \tag {*}    \\
			E \es{z}{xy} E_v & \mtwo E \es{z}{E_v(y'_1)}E_v \tag{\#} \\ 
		\end{align*}
    \vspace{-8ex}
    \[ \to_{Cr} = \to_{m_1} \cup \to_{m_2} \]
		\vspace{-3ex}
		\begin{alignat*}{5}
                  & (*)\   &  & E_v(x)^\alpha &&= \lambda x_1. \ess b E \quad &  & (\ess{b} E) \sub{x_1}{y} &&= \ess {b'} E' \\
                  & (\#)\  &  & E_v(x) &&= \lambda x_1. \ess{y_1} \qquad         &  & \ess{y_1} \sub{x_1}{y} &&= \ess {y'_1}
		\end{alignat*}
    \

	\end{mdframed}
	\caption{Finest crumbled calculus}
	\label{fig:crumbled-calculus}
\end{figure}

The idea behind crumbling consists in avoiding the nesting of applications, so that we can write terms such as $(tu)s$ and $t(us)$ as $(xs)\es{x}{tu}$ and $(tx)\es{x}{us}$ by using explicit substitutions: this is the most common way for representing the crumbling, and represents the most idiomatic way for representing environment entries.
Many forms of crumbling have been presented, from more coarse 
in \cite{accattoliCrumblingAbstractMachines2019}, to finer ones in \cite{accattoliPositiveSharingAbstract2025}.
The form we present -- that we call {\em \crumbling\ }-- is an even finer form of crumbling where an explicit substitution may only appear either as an application of variables or as a $\lambda$-abstraction.

We add to this representation another ingredient: the concept of \textit{returning a value} by explicitly assigning it to the special variable $*$ via an explicit substitution. This is reminiscent of Pascal where the return value is the one assigned to the function name in the body of a function. Thanks to $*$ we can get rid of terms altogether and just work with \textit{environments} that are sequences of explicit substitutions. The whole term to be reduced is thus turned into an environment and the same applies to the bodies of abstractions.
We present the concept with the following example.
\begin{example}
  Consider the closed term $t := \lambda x. \lambda y. (xx) (yy)$: by using the fresh variables $z_1, z_2$, the finest crumble representation of the term is
  \[ \ess{\lambda x. \ess{\lambda y. \ess{z_1 z_2} \es{z_1}{x x} \es{z_2}{y y}}} \]
\end{example}
Note that assignments to $*$ only make sense in the leftmost position of an environment.

By the reduction rules defined in \autoref{fig:crumbled-calculus}, the crumbling is stable, i.e. a crumble reduces to another crumble.
The only tricky point will be given by representing and evaluating the identity and constant functions, i.e. $\lambda x. x,\ \lambda x. y$ as
they can only be represented assigning a variable to $*$, which is only allowed in this ad-hoc case since
a variable is neither an abstraction nor an application of a variable to a variable.

The original motivation for crumbling is that it trivializes the search for the next redex, as we visit the entries of the environment in a right-to-left order until we find the redex, independently by the particular strategy (right-to-left vs left-to-right) that we want to capture. 
Consider the following Plotkin term evaluated in a right-to-left order, where $v_1$, $v_2$ and $v_3$ are some values. The colour $\textcolor{blue}{blue}$ indicates the term that is already evaluated, the colour $\textcolor{red}{red}$ indicates the part of the term that is yet to be evaluated, and the colour $\textcolor{green}{green}$ indicates the next term to be evaluated.
\[ 
  \textcolor{red}{(II)}
  \fcolorbox{green}{white} {$\textcolor{blue}{(v_1 v_2)}$} 
  \textcolor{blue}{v_3} 
\]
The right-to-left crumbled representation of the term would be the following 
\[ 
   \textcolor{red}{\ess{z_1 z_2}\!\es{z_1}{\crumb{I}}\!\es{z_2}{z_3 z_4}\!\es{z_3}{\crumb{I}}\!\es{z_4}{z_5 z_6}}
   \textcolor{green}{\es{z_5}{z_7 z_8}} 
   \textcolor{blue}{\es{z_7}{\crumb{v_1}}\!\es{z_8}{\crumb{v_2}}\!\es{z_6}{\crumb{v_3}}}
\]
On the other hand, the same Plotkin term evaluated with a left-to-right strategy would be "coloured" as follows
\[ 
  \fcolorbox{green}{white} {$\textcolor{blue}{(II)}$} 
  \textcolor{red}{(v_1 v_2) v_3} 
\]
Whereas its left-to-right crumbled representation would be
\[ 
  \textcolor{red}{ \ess{z_1 z_2}\!\es{z_2}{z_3 z_4}\!\es{z_4}{\crumb{v_3}}\!\es{z_3}{z_5 z_6}\!\es{z_6}{\crumb{v_2}}\!\es{z_5}{\crumb{v_1}} }
  \textcolor{green}{ \es{z_1}{z_7 z_8} }
  \textcolor{blue}{\es{z_8}{\crumb{I}}\!\es{z_7}{\crumb{I}} }
\]
It is evident then that the strategy adopted to search the next $\beta$-redex is induced by the crumbling process, as it is embedded in the ordering of the concatenation of the environment.
The net effect is that reduction machines for crumbled calculi, that we will introduce later, do not need applicative stacks nor dumps or other additional data structures that are required to implement the search for the next redex in the scenario of non crumbled calculi.

\paragraph{Syntax: crumbled forms and environment contexts.}
Terms from $\lambda_{Plot}$ are replaced by crumbles, which take the form of an environment.
An environment is defined as a sequence of explicit substitution of bites for variables; a bite is either a crumbled value or an application of two variables. 
\begin{definition}[Crumble]
  A crumble $E$ is an environment where the leftmost explicit substitution is in the form $\ess{b}$, i.e. $E = \ess{b}E'$
\end{definition}

We also introduce \textit{environment contexts} for which the plugging operation is defined as $C\plug{b} = (E \es{z}{\hole})\plug{b} = E \es{z}{b}$.
A crumbled value is either the crumbling of an identity or constant function, or the crumbling of an abstraction whose body is a generic term. 
Note that this distinction is needed for ensuring that every term has a corresponding crumble.
For example, $\lambda x. y$ has crumbled form $\ess{\lambda x. \ess{y}}$: obviously, this crumble could not be represented with the rule we give for generic abstractions, so we need to explicitly define syntax and rules for these terms.

We postpone the discussion on the evaluation of crumbled forms to \autoref{sect:creval} after the presentation of the mutual translations between $\lambda_{Plot}$-terms
and their corresponding crumbled forms.

\subsection{The \crumbling\ translation and its inverse}
Let $\Lambda$ and $\crumb{\Lambda}$ denote respectively the set of all $\lambda_{Plot}$-terms and the set of all finest crumbled forms.
We define the crumbling translation $\crumb{\cdot} : \Lambda \to \crumb{\Lambda}$ in order to translate a term from the $\lambda_{Plot}$ calculus to its crumbled representation.
Together with this translation function we define the \textit{read-back} function $\rback{\cdot} : \crumb{\Lambda} \to \Lambda$ that is the inverse function of the crumbling, used to translate a crumble back to its $\lambda_{Plot}$ representation.
The formal definitions of the functions are summarized in \autoref{fig:crumbled-forms}. 
Note that the read-back of an environment context may not always be a context: as a simple example, let $C := \ess{\lambda x. \ess{x x} \es{x}{\hole}}$ then $\rback{C} = \lambda x. \hole \hole$ that is not a context. We will however provide some invariants so that all environment contexts we consider will read-back to Plotkin's right v-contexts.

\begin{figure}[!t]
	\begin{mdframed}
		\centering
		\textsc{Crumbling}
		\begin{gather*}
			\crumb{xy} = \ess{xy}
			\qquad \crumb{\lambda x. y} = \es * {\lambda x. \es * y}
			\qquad \crumb{\lambda x. t} = \ess{\lambda x. \crumb{t}} \\
		\end{gather*} \vspace{-4em}
    \begin{alignat*}{2}
      \crumb{uy} & = \es * {x y} \es{x}{b} E \quad
      && \text{ with $x$ fresh, } \crumb{u} = \ess{b}E \\
            \quad \crumb{xu'} & = \es * {xy} \es {y} {b'} E' \quad
      && \text{ with $y$ fresh,  } \crumb{u'} = \ess{b'}E' \\
            \quad \crumb{uu'} & = \ess{x y} \es{x }{b}E \es{y}{b'}E' \quad 
      && \text{ with $x$, $y$ fresh,  } \crumb{u} = \ess{b}E, \crumb{u'} = \ess{b'}E'
		\end{alignat*}
		\hrule
		\vspace{1ex}
		\textsc{Read-back}
		\begin{gather*}
			\rback{ (xy) } = xy
			\qquad \rback{(\lambda x. \es * y)} = \lambda x. y
			\qquad \rback{(\lambda x. E)} = \lambda x. \rback E \\
			\rback{ \es * b } = \rback b
			\qquad \rback{ (E \es x b) } = \rback E \sub{x}{\rback{b}} \\
			\rback{\ess{\hole}} = \hole
			\qquad \rback{(E\es{z}{\hole})} = \rback{E} \sub{z}{\hole}
		\end{gather*}
	\end{mdframed}
        \caption{Translation from $\lambda_{Prot}$-terms to crumbled forms and back}
	\label{fig:crumbled-forms}
\end{figure}

\begin{figure}[!t]
  \small
  \begin{mdframed}
    \centering
    \textsc{Size of Terms}     
    \begin{tabularx}{\textwidth}{CCC}
      $\size{x} = 1$
      & $\size{\lambda x. t} = \size{t} + 1$
      & $\size{uu'} = \size{u} + \size{u'} + 1$
    \end{tabularx}

    \vspace{1ex}
    \textsc{Size of bites}
    \begin{tabularx}{\textwidth}{CCC}
      $\size{\lambda x. \ess{y}} = 2$
      & $\size{\lambda x. E} = \size{E} + 1$
      & $\size{xy} = 2$
    \end{tabularx}

    \vspace{1ex}
    \begin{tabular}{l ll@{\quad\quad}lll}
    \textsc{Size of Environments}
      \qquad & $\size{\epsilon}$ & $= 0$
      \qquad & $\size{E \es{z}{b}}$ & $= \size{E}$ & $ + \size{b}$ \\

    \textsc{Length of Environments}
      \qquad & $\len{\epsilon}$ & $= 0$
      \qquad & $\len{E \es{z}{b}}$ & $= \len{E}$ & $ +\, 1$
    \end{tabular}

    \vspace{1ex}
    \hrule
    \vspace{.5ex}
    \centering
    \textsc{Domains}

    \begin{tabular}{ll}
      $\dom(\epsilon)$ 
        & $ = \emptyset$ \\
      $\dom(E \es{x}{b})$
        & $ = \dom(E) \cup \{x\}$
    \end{tabular}

    \textsc{Free Variables}

    \begin{tabular}{ll}
      $\fv(\epsilon)$
        & $ = \emptyset$ \\
      
      $\fv(E\es{z}{\lambda x. E'})$
        & $ = \fv(E)\!\setminus\!\{z\} \cup (\fv(E')\!\setminus\!\{x\})$ \\
      
      $\fv(E \es{z}{xy})\!$ 
        & $ =\!\fv(E)\!\setminus\!\{z\} \cup \{x,y\}$ \\

      $\fv(E\es{z}{\lambda x. \ess{y}}) $ 
        & $= \fv(E)\!\setminus\!\{z\} \cup (\{y\}\!\setminus\!\{x\})$
    \end{tabular}

    \textsc{Bound Variables}

    \begin{tabular}{ll}
      $\bv(\epsilon)$
        & $ = \emptyset $ \\

      $\bv(E\es{z}{\lambda x. E'})$ 
        & $ = \bv(E) \cup \{x\} \cup \bv(E')$ \\
      
      $\bv(E \es{z}{xy})$
        & $ = \bv(E)$ \\
      
      $\bv(E\es{z}{\lambda x. \ess{y}})$ 
        & $ = \bv(E) \cup \{x\}$
    \end{tabular}
  \end{mdframed}
  \caption{Measures for terms, bites and environments; bound and free variables, domains}
  \label{fig:size-and-length}
\end{figure}

\paragraph{Properties of crumbling.}
We now define some notions that characterize the properties of crumbling; we will also show that these properties become invariants of the reduction on crumbles.

In the top half of~\autoref{fig:size-and-length} we define the measures of size for terms and environments, together with the definition of length of an environment, that simply counts the number of entries it contains.

The definitions of free variables and bound variables are extended from $\lambda_{Plot}$ to environments with the auxiliary notion of \textit{domains}. The formal definitions are shown in the bottom half of~\autoref{fig:size-and-length}.

\begin{definition}[Well-named crumbles]A crumble $E$ is \emph{well-named} if all the variables appearing in $\dom(E)$ are pair-wise distinct and $\dom(E) \cap \bv(E) = \emptyset$ and if $E = E_1 \es{z}{b} E_2$ then $z \notin \fv(b) \cup \fv(E_2)$.
\end{definition}

\begin{restatable}[Properties of Crumbling]{lemma}{crumblingprop}
	For every closed term t:
  \label{lem:crumb-properties}
	\begin{enumerate}
		\item Freshness: $\crumb{t}$ is well-named.
		\item Closure: $\crumb{t}$ is closed.
		\item Bodies: each body in $\crumb{t}$ is the translation of a term.
    \item Size: the size $| \crumb{t} |$ is linear with respect to the size of the initial term $t$. \label{lem:crumb-size}
		\item Contextual decoding: if $\crumb{t} = E_1 \es{z}{xy} E_2$ then $\rback{(E_1 \es{z}{\hole})}$ is a right v-context.
	\end{enumerate}
\end{restatable}

\begin{figure}[!t]
  \begin{mdframed}
    \centering
    \textsc{Example}
    \begin{align*}
      & \ess{z_1 z_2} \es{z_1}{\lambda x. \ess{x z_3} \es{z_3}{x x}} \es{z_2}{\lambda y. \ess{y}} \\
      \mone & \ess{z_2 z_4} \es{z_4}{z_2 z_2} \es{z_1}{\lambda x. \ess{x z_3} \es{z_3}{x x}} \es{z_2}{\lambda y. \ess{y}} \\
      \mtwo & \ess{z_2 z_4} \es{z_4}{\lambda y. \ess{y}} \es{z_1}{\lambda x. \ess{x z_3} \es{z_3}{x x}} \es{z_2}{\lambda y. \ess{y}} \\
      \mtwo & \ess{\lambda y. \ess{y}} \es{z_4}{\lambda y. \ess{y}} \es{z_1}{\lambda x. \ess{x z_3} \es{z_3}{x x}} \es{z_2}{\lambda y. \ess{y}}
    \end{align*}
    
  \end{mdframed}

  \caption{Example of evaluation of the crumbling of the term $(\lambda x. x (x x)) \lambda y. y$}
  \label{fig:example-cr-eval}
\end{figure}

\vspace{-0.25cm}

\subsection{Evaluating crumbles}\label{sect:creval}
The reduction rules for the finest crumbled calculus are also given in~\autoref{fig:crumbled-calculus}.
All environment marked $E_v$ in the paper are required to be \textit{v-environment}:
\begin{definition}[v-environment]An environment is a $v-environment$ if
each entry maps a variable to a value, i.e. is in the form $\es{x}{v}$.
\end{definition}
The use of $E_v$ in the rules forces evaluation to proceed in a right-to-left order and therefore the longest v-environment suffix of the environment represents the already fully-evaluated part.

By $E_v(x)^\alpha$ we denote the crumbled value obtained from $E_v(x)$ by $\alpha$-renaming the variables in the domain of $E_v(x)$ with fresh ones, so that the resulting concatenation of environments is well-named.

The reduction rules are triggered either if the abstraction is an identity/constant function (in which case $m_2$ is triggered) or if it is a more generic function (in which case $m_1$ is triggered).
The rule $m_1$ substitutes the $\beta$-redex with the bite (on which is performed the immediate substitution $\sub{x_1}{y}$) returned by the $\lambda$-abstraction and concatenates the tail of the environment contained in the body of the abstraction mapped to the variable $x$. 
Since the rule $m_2$ is triggered when the abstraction is an identity or constant function, it substitutes the $\beta$-redex with the value mapped to variable of the body on which is performed the immediate substitution $\sub{x_1}{y}$ and does not need to perform concatenation.
\autoref{fig:example-cr-eval} shows an example of how evaluation works.
\begin{definition}[Reachable Crumble]
  A crumble is reachable if it is obtained by a sequence of $\to_{Cr}$-steps starting from the translation $\crumb{t}$ of a closed term $t$.
\end{definition}
\begin{definition}[v-crumble]
  A v-crumble is a v-environment that is also a crumble, i.e. its in the form $\ess{v}E_v$.
\end{definition}
\begin{restatable}[Harmony for \crumbling]{lemma}{crumblingharmony}
	\label{lem:crumble-harmony}
  A reachable crumble $E$ is normal if and only if it is a v-crumble.
\end{restatable}
The next Lemma shows that the properties of the \crumbling\ shown in \autoref{lem:crumb-properties} are preserved during the evaluation.

\begin{restatable}[Invariants of \crumbling]{lemma}{crumblinginvariants}
  \label{lem:crumb-invariants}
  For a reachable crumble $E$:
	\begin{enumerate}
    \item Freshness: $E$ is well-named. \label{lem:invariant-well-naming}
    \item Closure: $\fv{(E)} = \emptyset$. \label{lem:invariant-close}
		\item Bodies: each body occurring in $E$ is a subterm (up to renaming and substitution of variable to variable) of the initial crumble.
    \item Contextual decoding: if $E = E_1 \es{z}{xy} E_2$ then $\rback{(E_1 \es{z}{\hole})}$ is a right v-context. \label{lem:invariant-c-decoding}
	\end{enumerate}
\end{restatable}
 
The bodies invariant is often referred to as the \textit{subterm invariant} \cite{accattoliCrumblingAbstractMachines2019,accattoliPositiveSharingAbstract2025} and is mostly used for the complexity analysis, as it forces the size of the terms duplicated along an evaluation to be bound by the size of the initial term, since the size is not affected by $\alpha$-renaming and variable substitution. Moreover, nobody has ever discovered so far a machine that does not exhibit this invariant but still has a reasonable complexity.
In our work we will also exploit the subterm invariant and the contextual decoding invariant to prove that the transitions $\mone, \mtwo$ project on evaluation steps \bvr in $\lambda_{Plot}$.

In the next section we show that evaluating the crumbles faithfully implements the semantics of $\lambda_{Plot}$, i.e. the \crumbling\ \textit{implements} $\lambda_{Plot}$.

\subsection{The Implementation Theorem}
In \cite{accattoliImplementingOpenCallbyValue2017} it is proven that, given
\begin{itemize}
  \item a calculus to be implemented, whose transition relation, made deterministic if necessary via an evaluation strategy, is $\to$
  \item an abstract machine $M$ given by a set of states and a transition relation $\leadsto_M$ that splits into
	      \begin{itemize}
		      \item principal transitions $\leadsto_p$, corresponding to evaluation steps on the calculus;
		      \item overhead transitions $\leadsto_o$, that take care of looking for the next redex and have no computational
                        meaning in the calculus
	      \end{itemize}
  \item a decoding $\rback{\cdot}$ of states of $M$ into terms
\end{itemize}
$M$ correctly implements $\to$ via $\rback{\cdot}$ whenever $(M, \to, \rback{\cdot})$ forms an \emph{implementation system}, i.e. whenever the following conditions hold true:
\begin{enumerate}
  \item Initialization: there is an encoding $\crumb{\cdot}$ of terms such that $\crback{t} = t$.
  \item Principal projection: $s \leadsto_p s'$ implies $\rback{s} \leadsto \rback{s'}$.
  \item Overhead transparency: $s \leadsto_o s'$ implies $\rback{s} = \rback{s'}$.
	\item Determinism: $\leadsto_M$ is deterministic.
	\item Halt: M final states (to which no transition applies) decode to $\to$-normal terms.
\item Overhead termination: $\leadsto_o$ terminates.
\end{enumerate}

As shown in \cite{accattoliImplementingOpenCallbyValue2017}, a machine implementation is such that there is a perfect match between the number of steps of the strategy and the number of principal transitions of the execution.

\begin{theorem}[Machine Implementation]
  If a machine $M$, a strategy $\to$ on $\lambda$-terms and a decoding $\rback{\cdot}$ form an implementation system, then:
  \begin{enumerate}
    \item Executions to derivations: for any $M$-execution $\rho: \crumb{t} \leadsto_M^* s$ there is a $\to$-derivation $d : t \to^* \rback{s}$.
    \item Derivations to executions: for every $\to$-derivation $d: t \to^* u$ there is an $M$-execution $\rho : \crumb{t} \leadsto_M^* s$ such that $\rback{s} = u$.
    \item Principal Matching: in both previous points the number $\size{\rho}_p$ of principal transitions in $\rho$ is exactly the length $\size{d}$ of the derivation $d$, i.e. $\size{d} = \size{\rho}_p$.
  \end{enumerate}
\end{theorem}

The finest crumbled calculus is not yet an abstract machine, but the implementation theorem is general enough to apply nevertheless.
In particular we can apply the theorem where the states of the machine are the reachable crumbles, the principal transitions are $\mone, \mtwo$ and
the set of overhead transitions is empty, since the discovery of the redex is embedded in the principal transitions. 

\begin{restatable}[Implementation System]{theorem}{crumblingimplementation}
  The \crumbling, Plotkin's weak CbV evaluation $\to_{\beta_v}$ and the read-back $\rback{(\cdot)}$ form an implementation system.
\end{restatable}
\begin{proof}
  We present just a rough sketch of the proof. For the full proof details see Appendix~\ref{app:proofs-crumbled-implementation}.
  \begin{enumerate}
    \item Initialization: by induction on the shape of $t$. It exploits the well-naming property (see \refitem{lem:crumb-invariants}{lem:invariant-well-naming}).
    \item Principal projection: by the contextual decoding property (see \refitem{lem:crumb-invariants}{lem:invariant-c-decoding}) the read-back of the left hand side of each principal rule is a redex, i.e. the application of a value to an abstraction in a right v-context.
    \item Determinism: to show determinism it suffices to show that the two multiplicative steps cannot be triggered simultaneously, but their side conditions are mutually exclusive.
    \item Halt: by closure property (see \refitem{lem:crumb-invariants}{lem:invariant-close}), if the crumble is normal then it is a v-crumble that reads-back to a value that is normal in Plotkin's CbV.
  \qed
  \end{enumerate}
\end{proof}

\section{The Reversible Crumbling Abstract Machine} \label{sec:rcam}
The finest crumbled calculus shown in the previous section solves the size-explosion issue with the use of explicit substitutions; however the search of the redex is embedded in the reduction rules, so the cost of each transition is not constant yet. 

We present the Reversible Crumbling Abstract Machine (RCAM) that enacts the Landauer's embedding through a history stack, and that presents an explicit transition for searching the next redex.
We then show that the complexity of the RCAM is bi-linear in the number of $\beta_v$-steps and the size of the initial term, that each transition of the machine is reversible and deterministic, and that the amount of additional space added to the history stack entry by each reduction rule is constant.

\subsection{The Machine} \label{sec:rcam-machine}
The Reversible Crumbling Abstract Machine (RCAM) is easily derived from the \crumbling\ calculus by enriching crumbles with an explicit pointer in the machine states that is used to visit the environment entries from right to left until a redex is found. Moreover a history stack is coupled with the crumble under evaluation to allow the reversibility of the transitions. Since the pointer already keeps track of where searching arrived, the entries in the history stack can be very concise: either the information that the pointer was moved, or the bite (of the form $xy$) that triggered the reduction step. The machine is formally defined in \autoref{fig:abstract_machine}.
The state of the machine is represented by three components: the \emph{active environment}, the \emph{evaluated environment} and the \emph{history stack}. The splitting of the environment into active and evaluated corresponds nicely on paper to having a pointer into an environment. Indeed the read-back from machine states to crumbles just concatenates the two, dropping the history. In realistic implementations the pair of active and evaluated environments form a zipper~\cite{zipper} where each cell (explicit substitution) points to the next one that is on the left for the active part and on the right on the evaluated part.
\begin{figure}[!t]
	\begin{mdframed}
		\centering
		\begin{tabularx}{\linewidth}{CC}
			\textsc{History Stack}
      & \textsc{State} \\
			$\mathcal H ::= \epsilon \mid \langle\rangle : \mathcal H \mid \langle x, y\rangle : \mathcal H$ 
      & $S ::= (E, E_v, \mathcal H)$ \\ [1ex]
			
      \textsc{Initialization}
      & \textsc{Read-back} \\
      $\iota(E) = (E, \epsilon, \epsilon)$
      & $\rbackm{(E, E_v, \hist)} = E E_v$
		\end{tabularx}
		\vspace{0.25em}
		\hrule
		\vspace{0.25em}
		\textsc{Transitions}
		{\small
      \renewcommand{\arraystretch}{2.5}
			\begin{tabular}{|c|c|c||c||c|c|c|c}
				\textsc{Act. Env}  
        & \textsc{Ev. Env}  
        & \textsc{Hist}
				& 
        & \textsc{Act. Env}    
        & \textsc{Ev. Env}     
        & \textsc{Hist}         
        & \\ \hline
				
        $E_1 \es{x}{v} $   
          & $ E_v$              
          & $H$                  
          & $\sealr$ 
          & $E_1$                 
          & $\es{x}{v}\!:\!E_v$ 
          & $\langle\rangle\!:\!H$                 
          & \\
				
        $E_1 \es{z}{xy}$ 
          & $E_v$ 
          & $H$ 
          & $\mmonelr$
          & $E_1 \es{z}{b'} E'$
          & $E_v$
          & $\langle x, y\rangle\!:\!H $ 
          & (*)  \\ [1ex]
				
        $E_1 \es{z}{xy}$
          & $E_v$ 
          & $H$ 
          & $\mmtwolr$ 
          & $E_1$
          & $\es z {E_v(y_1')}E_v$
          & $\langle x, y\rangle\!:\!H$
          & (\#)
			\end{tabular}}

		\begin{alignat*}{4}
                  & (*)\   &  & E_v(x)^\alpha &&= \lambda x_1. \es * b E \quad (\es * b E) \sub{x_1}{y} &&= \es * {b'} E' \\
                  & (\#)\  &  & E_v(x) &&= \lambda x_1. \ess{y_1} \quad \ess{y_1} \sub{x_1}{y} &&= \ess{y_1'}
		\end{alignat*}
    \begin{tabularx}{\textwidth}{CC}
      $\rcr =\ \mmoner \cup \mmtwor \cup \sear $
      & $\rcrinv =\ \mmonel \cup \mmtwol \cup \seal$ \\
    \end{tabularx}
      $\rcrfull = \rcr \cup \rcrinv$
	\end{mdframed}
  \caption{Reversible Crumbling Abstract Machine}
	\label{fig:abstract_machine}
\end{figure}

\paragraph{The transitions.} \label{par:rcam-transitions}
It is easy to see that the forward transitions $\mmoner, \mmtwor$ are exactly the $\mone, \mtwo$ transitions of the \crumbling\ calculus once read-back is applied and that
the $\sear$ is the identity up to read-back. In \autoref{sec:rcam-reversibility} we will deeply discuss the reversibility of each transition.
\begin{definition}[Reachable State]
  A state is reachable if it is obtained by a sequence of evaluation steps in $\rcrfull$ starting from the initialization $\initm{E}$ of a closed, well-named crumble $E$.
\end{definition}
The harmony property extends also to the RCAM. Since all machine rules require the evaluated environment to be a v-environment, proving harmony entails proving that in every reachable state the evaluated environment is always a v-environment. 
Therefore a concrete implementation of the machine can avoid to check whether the environment is indeed a v-environment.

\begin{restatable}[Harmony for the RCAM]{lemma}{harmonyrcam}
	\label{lem:harmony-rcam}
  A reachable by forward-only transitions state $S$ is $\rcr$-normal if and only if it is in the form $S = (\epsilon, E_v, \hist)$.
\end{restatable}

\begin{restatable}[Invariants for the RCAM]{lemma}{invariantsrcam}
  \label{lem:rcam-invariants}
  Let $S$ be a reachable by forward-only transitions state of the RCAM:
  \begin{enumerate}
    \item Freshness: $S$ is well-named.
    \item Closure: $S$ is closed.
    \item Rightmost: $S = (E, E_v, \hist)$ for some environment $E$, v-environment $E_v$ and history stack $\hist$.
  \end{enumerate}
\end{restatable}

\paragraph{The Implementation Theorem.}
We prove now that, when restricted to forward transitions only, the RCAM correctly implements
the finest crumbled calculus. We exploit once again the implementation theorem, this time also
considering overhead transitions.
\begin{restatable}[Implementation Theorem for the RCAM]{lemma}{rcamimplementation}
  \label{lem:impl-rcam}
  The RCAM, $\to_{Cr}$ evaluation and $\rbackm{(\cdot)}$ form an implementation system.
\end{restatable}
\begin{proof}
  The proof can be found in Appendix~\ref{app:proofs-rcam-rev}.
  The only non trivial part is overhead termination: we will show in Section~\ref{sec:rcam-complexity} that, when restricted to forward transitions only, the number of $\sear$ transitions is bi-linear in the number of principal transitions and the maximal length of an environment in the initial crumble, so the $\sear$ transition terminates. \qed
\end{proof}

\subsection{Reversing the transitions} \label{sec:rcam-reversibility}
The new (backward) transitions $\mmonel, \mmtwol, \seal$ are the reverse of $\mmoner$, $\mmtwor$, $\sear$ respectively, and they exploit the history stack.
From \autoref{fig:abstract_machine} we can see that the history stack will guide us in reversing each evaluation step, as we only have two scenarios:
\begin{itemize}
	\item If the top of the stack is $\langle\rangle$, then the evaluation step was an overhead search transition, so we just need to move the pointer to the next cell on the right.
        \item If the top of the stack is $\plug{x, y}$, since we have saved $x$ on the stack, we know exactly the value $E_v(x)$. From the shape of $E_v(x)$ we also understand if the step to be performed is $\mmonel$ or $\mmtwol$. For the case $\mmonel$ only, from $E_v(x)$ we know the length of $E'$, as $\len{E} = \len{E'}$, and we can remove $\len{E'}$ entries from the environment. Then in both cases we turn the substitution for $z$ into $\es{z}{xy}$ and, in case of $\mmtwol$, shift it back to the active environment.
\end{itemize}
\paragraph{Size of the history entries}
Note that the amount of information we save on the history stack is reasonably minimal, as we need to distinguish which transition led to the current state of the machine. Consider the following state of the machine where we have removed the history stack component
\[ (\epsilon, \ess{\lambda x. \ess{x}} \es{z_1}{\lambda x. \ess{x}} \es{z_2}{\lambda x. \ess{x}}) \]
We may reach this state in five different ways:
\begin{enumerate}
  \item With a transition $m_2$ reducing the redex $y_1 y_2$, with the 4 possible combinations of $y_1, y_2 \in \{z_1, z_2\}$ i.e. starting from one of the possible states $(\ess{y_1 y_2},  \es{z_1}{\lambda x. \ess{x}} \es{z_2}{\lambda x. \ess{x}})$ for $y_1,y_2 \in \{z_1,z_2\}$.
  \item With a transition $sea$ if the machine still needs to evaluate the last entry, i.e. starting from the state $(\ess{\lambda x. \ess{x}}, \es{z_1}{\lambda x. \ess{x}} \es{z_2}{\lambda x. \ess{x}})$.
\end{enumerate}
To distinguish the four scenarios described by the possible combinations of $z_1 z_2$, we need to save on the stack both variables in the bite associated to $*$; moreover, to distinguish the last scenario, we need to save something different on the stack: this is the role of $\langle\rangle$ in~\autoref{fig:abstract_machine}.
Note that the amount of information we store on the history stack is small enough to be acceptable: the concrete implementation shown in Appendix~\ref{app:implementation} represents an explicit substitution $\es{x}{b}$ with a record in the heap that holds $b$ together with a pointer to the previous substitution in the environment. In such an implementation the occurrences of the variable $x$ are just pointers to the memory address of the record, i.e. variables are pointers and thus can be stored in $O(1)$ space. Therefore each history stack entry uses size $O(1)$ as well holding at most two pointers (the two variables).

\paragraph{Properties of reversible transitions and computations}
We now identify some properties of the reverse transitions of RCAM. The terminology is taken
from~\cite{laneseAxiomaticApproachReversible2020,danosReversible2004} when applicable. Note, however, that since our
language is sequential most properties of~\cite{laneseAxiomaticApproachReversible2020} hold trivially
and therefore we skip them.

We define an \textit{execution} $\rho: S_0 \rcrfull^* S$ as a sequence of RCAM transitions that start in the initial state $S_0$ and end in $S$. 
A state $S$ is said to be reachable through an execution $\rho$ if $\rho: S_0 \rcrfull^* S$.

\begin{restatable}[Backward determinism]{lemma}{backwarddet}
  \label{lem:revbackward}
  The transitions $\rcrinv$ are deterministic.
\end{restatable}

\begin{restatable}[Loop]{corollary}{invertstate}
  \label{lem:invertstate}
  Let $S'$ be any state of the RCAM, such that $S \rcr S'$. Then the evaluation step can always be reverted, i.e. $S' \rcrinv S$ and vice-versa.
\end{restatable}

\begin{restatable}{corollary}{idrev}
  \label{lem:id-rev}
  The sequence of a forward transition with a backward one represents the identity, i.e. if $\rho : S_0 \rcr S' \rcrinv S$ then $S_0 = S$.
\end{restatable}

\begin{restatable}{lemma}{fwdonly}
  \label{lem:rev-fwd-only}
  Given an execution $\rho$ such that a state $S'$ is reachable from an initial state $S_0$ through $\rho$, i.e. $\rho: S_0 \rcrfull^* S'$, then $\exists \rho' : S_0 \rcr^* S'$. 
\end{restatable}

\begin{restatable}[Well-Foundedness]{theorem}{welfoundedness}
  There is no infinite reverse computation, i.e. we do not have sequences $(S_i)_{i\in\mathbb{N}}$ such that $S_{i+1} \rcr S_i$ for all $i \in \mathbb{N}$.
\end{restatable}

\begin{restatable}{theorem}{harmonyrev}
  A reachable state $S$ is $\rcrinv$-normal if and only if it is initial.
\end{restatable}

\subsection{Complexity Analysis} \label{sec:rcam-complexity}
We target now the cost analysis of the machine when restricted to forward rules only, with the aim of expressing the cost as a function of the number of principal transitions only. Of course we rule out backward rules otherwise the machine could perform and undo a search rule an unbounded number of time before doing the next principal rule.

Following \cite{accattoliCrumblingAbstractMachines2019}, the cost analysis starts by bounding the number of search transitions in an execution $\rho$. Then,
after discussing the cost of implementing single transitions, we will compose the two analyses to obtain the total cost, showing that the RCAM is bi-linear in the size of the initial crumble and the number of
principal steps. Finally we will show that the cost of each backward transition is identical to that of
its forward counterpart.

\paragraph{Number of search transitions.}
For the complexity analysis of the RCAM we will make use of the measure of size defined in \autoref{sec:crumbled-calculus}. We define the measure of length of a state of the RCAM as the length of its read-back, i.e. $\len{S} = \len{\rbackm{S}}$.
Since after a multiplicative step the measure increases by the length of the body we need to concatenate, we also define the function $L(E)$ that bounds the length of the bodies occurring in $E$.
\[ L(E) = \sup{\{|E_b|_{len} : E_b \text{ is a body appearing in } E}\} \]
The function can be extended to the states of the RCAM simply by using the read-back of the state.
\begin{restatable}{lemma}{lenrcam}
  \label{lem:lenrcam}
  Let $E$ be a well-named crumble and let $\rho : \initm{E} \rcr^* S$ an execution in the RCAM with $S = (E_1, E_v, \hist)$. Then $\len{E_1} \leq \len{E} + \size{\rho}_p \cdot L(E) - \size{\rho}_{sea}$.
\end{restatable}

\begin{restatable}[Number of Search transitions]{corollary}{numsea}\label{lem:count-sea}
  Let $t$ be a closed term. For a normalizing forward steps only $\rho$ in the RCAM starting from $\iota{(\crumb{t})}$, we have $|\rho|_{sea} \leq (|\rho|_p + 1) \cdot |t|$.
\end{restatable}

\paragraph{Cost of forward transitions.}
Since the read-back of a state of the machine returns a crumble, we will address the cost of transitions by talking in terms of crumbles and environments. Updating the history stack has constant cost and so it will not be part of the discussion. Computing $E_v(x)$ as well as appending environments can be implemented in $O(1)$ on a Random Access Machine as shown in~\cite{accattoliCrumblingAbstractMachines2019} and Appendix~\ref{app:implementation}.

The cost of each principal forward transition is bound by the size of the crumble in the abstraction that needs to by copied for $\alpha$-renaming it and applying the variable-to-variable substitution. By the bodies invariant, the abstraction is the $\alpha$-renaming of an abstraction present in the initial crumble. 
By \autoref{lem:crumb-properties}.\ref{lem:crumb-size}, the size of the initial crumble is linear in the size of the initial term, therefore the cost of a principal forward transition is linear in the size of the initial term. The cost of $\sear$ is clearly $O(1)$.

\paragraph{Cost of the execution.}
By combining the analysis from the previous paragraphs, we obtain the following theorem

\begin{theorem}[The RCAM is bi-linear]
  For any closed, term $t$ and any RCAM execution $\rho: \initm{\crumb{t}} \rcr^* S$, the cost of implementing $\rho$ on a RAM is $O((|\rho|_p + 1) \cdot |t|)$.
\end{theorem}

\begin{theorem}[Space overhead of Landauer's embedding]
  For any closed, term $t$ and any RCAM execution $\rho: \initm{\crumb{t}} \rcr^* S$, the maximum size of the history during the execution $\rho$ on a RAM is $O((|\rho|_p + 1) \cdot |t|)$.
\end{theorem}

\begin{proof}
  The space overhead of the embedding is given by the history, that grows monotonically during any execution, adding an $O(1)$ entry for each principal and search step.
\end{proof}

\paragraph{Cost of backward transitions.}
The costs of the backward transitions are identical to their forward counterpart.
For $\seal$ we just need to shift a substitution in a Zipper: as for $\sear$, the operation costs $O(1)$.
The case of $\mmtwol$ is similar and costs $O(1)$ as well, exactly as $\mmtwor$.
For $\mmonel$, we need to remove the $|E'|_{len}$ substitutions that were concatenated during $\mmoner$: the cost for the single transition is then the same.
For $\mmonel$ we could actually achieve a better result if we saved on the history stack also the variable $z$, as this would make the cost of $\mmonel$ $O(1)$, leading however to a trade-off between space and time.

\section{Conclusions and future works}
To the best of our knowledge we have presented the first Landauer's embedding for Plotkin's (weak, closed) CbV calculus that requires
constant history overhead. The key idea was to made reversible a not yet explored variant of Crumbled Abstract Machine, a proven technology to implement call-by-value calculi with the best known asymptotic cost. We firmly believe that our approach scales seamlessly to other constructs like \texttt{if-then-else} or pattern matching over algebraic data types following the same pattern. For example, to accommodate boolean and \texttt{if-then-else} similarly to~\cite{accattoliCrumblingAbstractMachines2019}, one would add the following bites to the grammar: \texttt{true} and \texttt{false}, that would also be values, and \texttt{if x then y else z} where $x,y,z$ are all variables; the only new history stack entry would be $\langle x,y,z \rangle$ to undo \texttt{if-then-else} reduction.

We also believe that the technique would scale as well to open call-by-value~\cite{accattoliOpenCallbyValue2016} and to strong call-by-value~\cite{accattoliStrongCallbyValueMulti2023}. The latter could find applications in the implementation of interactive theorem provers where sometimes reduction is used during speculative proof-search, leading to the need for backtracking.

We conjecture the possibility to find cheap Landauer's embedding for some classes of reduction machines used to implement efficiently call-by-name and the much harder call-by-need, especially in the strong setting. However, the details could differ significantly from what is presented in this paper since the machines for call-by-need are more complicated.

Several reversible functional programming languages have been proposed in the literature, like~\cite{yokoyamaReversibleFunctionalLanguage2012}. Unlike the non reversible case, where Plotkin's CbV represents the core of the language, these languages are not extension of Plotkin's CbV. In particular they tend to exhibit a clear distinction between first-order data and reversible functions. Therefore we do not see any direct correlation between our contribution and that part of the literature.

An easy application of our work is to debugging. Indeed reversible debuggers like~\cite{ocamldebugger} have been available and appreciated since a long time. They are typically implemented using snapshots of the state that are taken regularly every $n$ computational steps. In order to undo a step, the computation goes back to the last snapshot and it is then replayed from there. If we consider a standard Crumbled Abstract Machine as an implementation of CbV, taking a snapshot requires recording all the active environment, that is subject to future in-place updates, and a pointer to the $v$-environment, that will not change. The size of the active environment, and thus of each snapshot, would be $O(|\rho|_p + 1)|t_0|$. Therefore the cumulative amount of space required by the snapshots would be quadratic in $|\rho|_p$ and linear in $|t_0|$, while our RCAM only requires bi-linear space.

A reference implementation in OCaml of the RCAM can be found at~\url{https://github.com/sacerdot/RCAM}. An explanation is given in~\autoref{app:implementation}. The implementation is meant for the reader to test the machine on some examples of her choice, to better understand the dynamics. Moreover it confirms that, even at low level, the complexity analysis of Subsection~\ref{sec:rcam-complexity} is correct.

\subsection{Related Works}
Other works have been presented on making abstract machines reversible in~\cite{klugeReversibleSEMCDMachine2000,huelsbergen1996logically}, with some space overhead.
As it is stated in~\cite{klugeReversibleSEMCDMachine2000}, the size of the overhead in~\cite{huelsbergen1996logically} is not minimal, and it is clearly not bi-linear.
Regarding the machine outlined in~\cite{klugeReversibleSEMCDMachine2000}, we believe that the size of the overhead should be linear in the number of steps, depending on the actual implementation model. However, no complexity analysis is presented to confirm this. In comparison to the RCAM, the SEMCD machine is significantly more complex, having 14 forward transitions an 8-components states, whereas the RCAM only has 3 transitions and 3-components states.

Another less related line of research is given by the Interaction Abstract Machine (IAM)~\cite{mackie1995geometry} and its derivation~\cite{accattoliMachineryInteraction2020} that are naturally reversible but implement a trade-off between space and time that is heavily unbalanaced versus space. Therefore, they don't attain the best bi-linear time complexity that our machine achieves.

\newpage
\bibliography{bibliography}
\bibliographystyle{splncs04}

\newpage
\appendix
\section{An implementation in OCaml of the RCAM}\label{app:implementation}

We provide at \url{https://github.com/sacerdot/RCAM}
a reference implementation in OCaml of the RCAM,
meant to guarantee the soundness of the assumptions that underly the complexity analyses.
The implementation also allows the reader to run the machine on examples of her choice, to better understand how the machine computes.

\begin{figure}[!t]
\includegraphics[width=\textwidth]{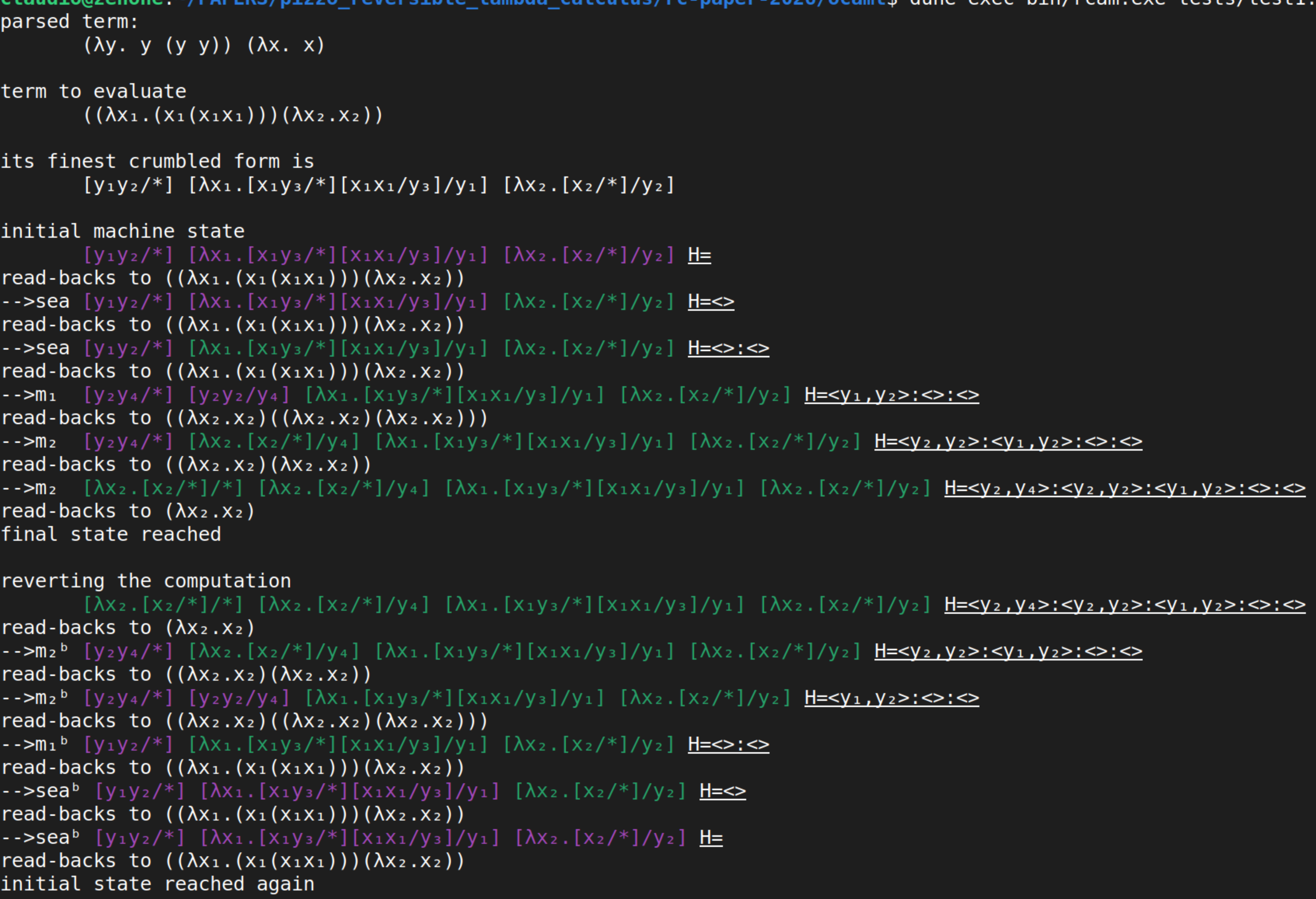}
\caption{An example of output}\label{img:run}
\end{figure}
The implementation takes as input a textual file whose content is a user provided term. The term is then crumbled and the crumble is shown to the user before starting a forward machine execution until the final state is reached. After each execution step both the machine state and its read-back are shown to the user. Once the final state is reached, the machine direction is reversed and the execution is shown until the initial state --- i.e. the final reverse state --- is reached again.  An example of output is shown in~\autoref{img:run}.

In the rest of the section we highlight a few significant details on the implementation.

\paragraph{Data structures.}
Bound variables (\lstinline{var}) are represented simply as heap-allocated cells: distinct cells in memory are distinct variables.
The \lstinline{mutable dummy : unit} code is an irrelevant trick to prevent the OCaml compiler to apply common subexpression elimination
and constant propagation, that would result in all variables being collapsed to a single cell in memory.
\begin{lstlisting}
(** Lambda bound variables **)
type var = { mutable dummy : unit }  (* only their memory address count;
                                        mutable is used to avoid common
                                        subexpression elimination *)
\end{lstlisting}

Plotkin's CbV terms are represented in the most obvious way using an Algebraic Data Type.
\begin{lstlisting}
(** Plotkin's terms **)
type term =
 | Var of var
 | Lam of var * term
 | App of term * term
\end{lstlisting}

Let's move now to crumbles. Explicit substitutions are implemented as \lstinline{node} that
are heap allocated cells that hold the bite to be substituted in the \lstinline{content} field.
In order to assemble explicit substitutions into environments, each cell also has a \lstinline{prev}
field to the previous node. The \lstinline{option} type is used to allow the first cell of the environment to
have no predecessor. The \lstinline{copying} field will be explained later when discussing how $\alpha$-renaming
is implemented. Since OCaml handles the \lstinline{node} type as boxed, occurrences of variables, i.e. OCaml values
of type \lstinline{node}, will be represented at runtime just as pointers into the heap and thus retrieving the
\lstinline{content} of a node is an $O(1)$ operation.
\begin{lstlisting}
(** Finest crumbled terms **)
type node =                       (* nodes are explicit substitutions *)
 { mutable content : bite         (* the substituted term *)
 ; mutable copying : node option  (* used only during copying to preserve sharing *)
 ; mutable prev : node_list       (* pointer to the next subst in the environment *)
 }
and node_list = node option       (* a list of nodes is an environment *)
\end{lstlisting}
Bites are implemented using an Algebraic Data Type that distinguish between the two shapes of abstractions used in the
paper. We introduce no explicit OCaml type for values. There are two additional changes w.r.t. the paper. The first one
is that we need to introduce a type of \lstinline{name} that collects together $\lambda$-bound variables and variables
bound by explicit substitutions, since the two are represented with different kinds of memory cells and since OCaml has
no notion of untyped pointer/reference. An implementation in C would not need this type. The second difference is that
an abstraction body $\es{*}{b}E$ is implemented as a \lstinline{crumbp} that is a pair made of $b$ and $E$ and, moreover,
the environment $E$ is represented by the \lstinline{env} datatype that holds both a pointer to the first and to the last
cell, unless the environment is empty. It would be possible to replace \lstinline{crumbp} with a \lstinline{node_list} to
remove the difference with the paper and the computational complexity would not change; however the code of some operations,
notably crumbling and $\alpha$-conversion would become less readable.
\begin{lstlisting}
and name = V of var | N of node   (* a name is either a lambda-bound variable or the address of a node *)
and bite =
 | Lam1 of var * crumbp
 | Lam2 of var * name
 | VV of name * name
and crumbp = bite * env           (* (b,e) encodes the crumble [b/*]e *)
and env = (node * node) option    (* to simplify the code we keep two pointers to the first
                                     and the last substitution in the environments used in
                                     the bodies of lambda-abstractions *)
\end{lstlisting}
History entries and machine states are implemented as expected. In particular the pair of active and evaluated nodes form a Zipper
and the $\seal$ and $\sear$ transitions become just shifting the Zipper to the left/right. We also introduce a datatype to
control the direction of execution of the machine.
\begin{lstlisting}
(** Reversible Crumbling Machine states **)
type history_entry =
 | Search
 | Principal of node * node

type machine_state =
 (* in a machine state (E1, Ev, H) the pair of environments
    (E1,Ev) implements a Zipper data structure; in particular
    cells in E1 point to their left and cells in Ev point to
    their right *)
 node_list * node_list * history_entry list

(** Reversible Crumbling Machine computation direction **)
type eval_direction = F | B       (* Forward vs Backward *)
\end{lstlisting}

\paragraph{Functions.}
The code for crumbling is given by a simple recursive function
\lstinline{translate : term -> node} that in 30 lines of code builds from a Plotkin's term
a non-empty environment, represented with (a pointer to) its last substitution.

The most tricky part of the code are the 30 lines of code that implement $\alpha$-conversion of environments. It is based on the standard algorithm to copy a graph in linear time preserving the sharing, described by Accattoli and Barras in \cite{accattoliEnvironmentsComplexityAbstract2017}. The \lstinline{copying : node option } field of a \lstinline{node} is used to temporarily associate during the copy process each node to its copy, if it has already been created. It differs from the standard algorithm because it also take cares of immediately substituting the argument
of the redex (a \lstinline{name}) for the abstracted variable (of type \lstinline{var}). The exact function name and type are \lstinline{copy_crumbp : var -> name -> crumbp -> crumbp} since the environments to be copied are bodies of abstractions and therefore they have the
special representation \lstinline{crumbp}.

The most important function is \lstinline{eval : eval_direction -> machine_state -> unit} that takes the direction of the computation and the machine state and recursively computes the normal form according to that direction, proceeding then at reverting the direction if it was the forward one. Adapting the function to just perform a step at a time and thus control the direction of the next step interactively would be a trivial change: just change all tail recursive calls to returning their argument. The code is the following. Note that checking the side
conditions of the principal rules is implemented very simply using deep patterns on the \lstinline{content} of nodes, i.e. by matching the result of the lookups in the $v$-environment.
\begin{lstlisting}
let rec eval dir ((n,ev,h) as state : machine_state) =
 print_endline (pretty_print_state state) ;
 print_string ("read-backs to " ^ string_of_t (read_back state) ^ "\n") ;
 match dir with
 | F ->
   (match n with
     | None ->
        (* normal form reached *)
        print_string "final state reached\n\n" ;
        print_string "reverting the computation\n       " ;
        eval B (n,ev,h)
     | Some n ->
        match n.content with
         | VV(N ({content=Lam1(y,c)} as v1), (N v2 as t)) ->
            step "m$\color{red}_1$ " ;
            let c2 = copy_crumbp y t c in
            eval dir (Some (n @ c2),ev,Principal(v1,v2)::h)
         | VV(N ({content=Lam2(y,z)} as v1), N({content=t} as v2)) ->
            step "m$\color{red}_2$ " ;
            n.content <- (subst_var y v2 z).content ;
            let n,ev = move_left n ev in
            eval dir (n,ev,Principal(v1,v2)::h)
         | Lam1 _ | Lam2 _ ->
            step "sea" ;
            let n,ev = move_left n ev in
            eval dir (n,ev,Search::h)
         | VV (_, V _) ->
            raise (Failure "term not closed!")
         | VV (N {content=VV _}, _) | VV (V _, _) ->
            (* non reachable state *)
            assert false)
 | B ->
   (match h with
     | [] ->
        (* Initial state reached going backward *)
        print_endline "initial state reached again"
     | Principal({content=Lam1(_,(_,env))} as v1,v2)::h ->
        step "m$\color{red}_1^b$" ;
        let len = length env in
        let m = drop len n in
        m.content <- VV(N v1, N v2) ;
        eval dir (Some m,ev,h)
     | Principal({content=Lam2(_,_)} as v1,v2)::h ->
        step "m$\color{red}_2^b$" ;
        let n,ev = move_right n ev in
        n.content <- VV(N v1, N v2) ;
        eval dir (Some n,ev,h)
     | Search::h ->
        step "sea$\color{red}^b$" ;
        let n,ev = move_right n ev in
        eval dir (Some n,ev,h)
     | Principal({content=VV _},_)::_ ->
        (* the input state is not reachable! *)
        assert false)

\end{lstlisting}

\section{Proofs for \autoref{sec:calculus} }\label{app:calculus}

\subsection{Proofs for \autoref{sec:plotkin-calculus}}
\begin{lemma}[Composition of right v-contexts]
	\label{lem:rvc-comp}
	Let $R, R'$ be two right v-contexts, then $R \plug{R'}$ is a right v-context.
\end{lemma}
\begin{proof}
  Trivial, see Lemma B.1 in \cite{accattoliCrumblingAbstractMachines2019}. \qed
\end{proof}

\subsection{Proofs for \autoref{sec:crumbled-calculus}}\label{app:proofs-crumbled}
\begin{lemma}
	\label{lem:lem1}
	Let $E_1\es{z}{xy}E_2$ be a closed crumble. Then $x, y \in \dom(E_2)$.
\end{lemma}
\begin{proof}
  It is sufficent to prove the most general statement that for all $x,b,E$, if $x \in fv(b)$ and
  $x \not\in fv(\es{z}{b}E)$ then $x \in dom(E)$. The proof is by structural induction on $E$.
  In the base case $E$ is empty and the two hypotheses contradict each other. In the inductive step
  $E = E'\es{y}{b'}$ and $fv(\es{z}{b}E'\es{y}{b'}) = fv(\es{z}{b}E') \setminus \{y\} \cup fv(b')$.
  If $x = y$ then $y \in dom(E)$ and the statement holds. Otherwise from the second hypotheses
  we deduce $x \not\in fv(\es{z}{b}E')$. By inductive hypothesis $x \in dom(E') \subseteq dom(E)$. 
  \qed
\end{proof}

\begin{lemma}
  \label{lem:norm-crumb-implies-v}
	If the closed crumble $E$ is normal, then it is a $v$-crumble.
\end{lemma}
\begin{proof}
	Suppose that $E$ is not a v-crumble, then $E = E_1\es{z}{xy}E_v$. However, since $E$ is closed, by \autoref{lem:lem1} $x \in \dom(E_v)$, and since $E_v$ is a v-environment, there will be an ES in $E_v$ of the form $\es{x}{v}$, but this would trigger a new reduction contradicting the hypothesis that $E$ is normal.
	\qed
\end{proof}

\begin{lemma}[Crumble commutes with renaming]
  The crumbling transformation commutes with the renaming of free variables, i.e. for any term $t$, $\crumb{t}\sub{z}{y} = \crumb{t \sub{z}{y}}$ if $z,y \notin \dom(\crumb{t})$.
	\label{lem:crumble-commute-free}
\end{lemma}
\begin{proof}
  By structural induction on $t$.

  If $t := \lambda z. t'$ then $t \sub{z}{y} = t$ and
  $\crumb{\lambda z.t'}\sub{z}{y} = \es{*}{\lambda z.\crumb{t'}}\sub{z}{y} = \crumb{\lambda z.t'}$; thus the property holds.

  If $t := \lambda x. z$ then $\crumb{t}\sub{z}{y} = \ess{\lambda x. \ess{z}}\sub{z}{y} = \ess{\lambda x. \ess{y}} = \crumb{t\sub{z}{y}}$.

  If $t := \lambda x. t'$ with $x \ne z$ then $\crumb{t} \sub{z}{y} = \ess{\lambda x. \crumb{t'}} \sub{z}{y}$, the claim follows by i.h. on $t'$.

  If $t := u u'$ then 
  \begin{alignat*}{3}
    \crumb{t}\sub{z}{y} 
    & = (\ess{x_1 x_2}\es{x_1}{b}E \es{x_2}{b'} E')\sub{z}{y} \\
    & = \ess{x_1 x_2}\sub{z}{y} (\es{x_1}{b}E)\sub{z}{y} (\es{x_2}{b'} E')\sub{z}{y} && \text{\quad by $z,y\notin \dom(\crumb{t})$} \\
    & = \ess{x_1 x_2}(\es{x_1}{b}E)\sub{z}{y} (\es{x_2}{b'} E')\sub{z}{y} \\
    & = \crumb{t \sub{z}{y}} && \text{\quad by i.h. on } u, u'
  \end{alignat*}

  If $t := ux$ then
  \begin{alignat*}{3}
    \crumb{t}\sub{z}{y} 
    & = (\ess{x_1 x} \es{x_1}{b}E) \sub{z}{y} \\
    & = \ess{x_1 x}\sub{z}{y} (\es{x_1}{b}E)\sub{z}{y} &&  \text{\quad by $z,y\notin \dom(\crumb{t})$} \\
    & = \ess{x_1 x} (\es{x_1}{b}E)\sub{z}{y} \\
    & = \crumb{t \sub{z}{y}} && \text{\quad by i.h. on } u, u'
  \end{alignat*}

  If $t := xu'$ the proof is analogous to the previous case.

  If $t := uz$ then
  \begin{alignat*}{3}
    \crumb{t}\sub{z}{y} 
    & = (\ess{x_1 z} \es{x_1}{b}E) \sub{z}{y} \\
    & = \ess{x_1 z}\sub{z}{y} (\es{x_1}{b}E)\sub{z}{y} &&  \text{\quad by $z,y\notin \dom(\crumb{t})$} \\
    & = \ess{x_1 y} (\es{x_1}{b}E)\sub{z}{y} \\
    & = \crumb{t \sub{z}{y}} && \text{\quad by i.h. on } u, u'
  \end{alignat*}

  If $t := zu'$ the proof is analogous to the previous case.
 \qed
\end{proof}

\begin{remark}[Crumble commutes with $\alpha$-renaming]
  \label{lem:crumble-commute-alpha}
  The crumbling transformation commutes with the the $\alpha$-renaming. \qed
\end{remark}

\begin{lemma}[Subterm Property]
	\label{lem:subterm-prop}
	Each body in a reachable crumble $E$ is a sub-term (up to renaming and substitution of variables to variables) of the initial crumble.
\end{lemma}
\begin{proof}
  We proceed by induction on the number of steps to reach $E$. In the base case $E$ is the initial term
  and the proof holds trivially. In the inductive step $E$ is the rhs of a reduction rule whose lhs
  is reachable in one step less. Therefore the inductive hypothesis holds for the lhs of the rule,
  i.e. all bodies in the lhs are sub-terms (up to renaming and substitution of variable to variable) of the initial crumble.
  We proceed by cases on the kind of rule.

  The rule $\mone$ copies and renames the body of an abstraction already present in the environment,
  then it substitutes on it a variable for another. Therefore all the new bodies, that are inside,
  have the property of interest.

  The rule $\mtwo$ copies the body of an abstraction already present in the environment,
  whose bodies have the property of interest. \qed
\end{proof}

\begin{corollary}
  \label{cor:body-is-a-crumble}
  Let $E$ be a reachable crumble, then every abstraction in $E$ is the finest crumbling of a term $t$.
\end{corollary}
\begin{proof}
  We distinguish between two cases.
  If the abstraction is $\lambda x_1. \ess{y}$, then trivially $E_b$ is the translation of the term $\lambda x_1. y$.
  Otherwise, the abstraction is $\lambda x_1. b$ where $b$ is a body that is not of the form $\ess{y}$ for some $y$.
  Therefore, by~\autoref{lem:subterm-prop}, $b$ is equal to $\crumb{t}$ up to renaming and substitution
  of variable to variable.
  By~\autoref{lem:crumble-commute-free} and~\autoref{lem:crumble-commute-alpha},
  since both renaming and substituting a variable to a variable commute with crumbling, there must be
  a $t'$ such that $b = \crumb{t'}$. Therefore $\lambda x_1.b = \lambda x_1.\crumb{t'} = \crumb{\lambda x_1.t'}$ and the statement holds.  \qed
\end{proof}

\begin{lemma}
  \label{rem:fv-term-crumb}
  For any $\lambda_{Plot}$-term $t$ that is not a variable, $\fv(t) = \fv(\crumb{t})$.
\end{lemma}
\begin{proof}
  We proceed by structural induction on $t$.
  \begin{itemize}
   \item Case $x$: impossible since $t$ is not a variable by hypothesis
   \item Case $\lambda x.y$ for $x,y$ not necessarily distinct:
     $\crumb{\lambda x.y} = \ess{\lambda x.\ess{y}}$ and
      $\fv(\lambda x.y) = \{y\} \setminus \{x\} = \fv(\ess{\lambda x.\ess{y}})$
    \item Case $\lambda x.t'$ for $t'$ not a variable:
      $\crumb{\lambda x.t'} = \ess{\lambda x.\crumb{t'}}$
      and by i.h. $\fv(t') = \fv(\crumb{t'})$. We have
      $$\fv(\lambda x.t') = \fv(t') \setminus \{x\} = \fv(\crumb{t'}) \setminus \{x\} =
        \fv(\ess{\lambda x.\crumb{t'}})$$
   \item Case $xy$: $\crumb{xy} = \ess{xy}$. We have
     $$\fv(xy) = \{x,y\} = \fv(\ess{xy})$$
   \item Case $uy$ for $u$ not a variable: $\crumb{uy} = \ess{xy}\es{x}{b}E$ where $\crumb{u}=\ess{b}E$ and thus, by i.h., $\fv(u) = \fv(\ess{b}E)$. We have
     $$\fv(uy) = \fv(u) \cup \{y\} = \fv(\ess{b}E) \cup \{y\} = \fv(\ess{xy}\es{x}{b}E)$$
   \item Case $xu'$ for $u'$ not a variable: $\crumb{xu'} = \ess{xy}\es{y}{b'}E'$ where $\crumb{u'}=\ess{b'}E'$ and thus, by i.h., $\fv(u') = \fv(\ess{b'}E')$. We have
     $$\fv(xu') = \fv(u') \cup \{x\} = \fv(\ess{b'}E') \cup \{x\} = \fv(\ess{xy}\es{y}{b'}E')$$
   \item Case $uu'$ for $u,u'$ not a variable: $\crumb{uu'} = \ess{xy}\es{x}{b}E\es{y}{b'}E'$ where $\crumb{u}=\ess{b}E$ and $\crumb{u'}=\ess{b'}E'$ and thus, by i.h., $\fv(u) = \fv(\ess{b}E)$ and $\fv(u') = \fv(\ess{b'}E')$. We have
     \begin{alignat*}{1}
         & \fv(uu')\\
       = & \fv(u) \cup \fv(u')\\
       = & \fv(\ess{b}E) \cup \fv(\ess{b'}E')\\
       = & \fv(\ess{xy}\es{x}{b}E\es{y}{b'}E') \tag*{\qed}
     \end{alignat*}
  \end{itemize}
\end{proof}

\begin{lemma}
  \label{lem:crumb-closed-prop}
	For any term $t$, if $t$ is closed then $\crumb{t}$ is closed.
\end{lemma}
\begin{proof}
  If $t$ is closed then it is not a variable and thus by~\autoref{rem:fv-term-crumb}
  $\emptyset = \fv(t) = \fv(\crumb{t})$ and thus $\crumb{t}$ is closed as well.  \qed
\end{proof}

\begin{lemma}
  \label{lem:closed-lem-1}
  For any closed term $t$, if $\crumb{t} = E_1 \es{z}{b} E_2$, then $b$ is closed in $E_2$, i.e. $\fv(b) \in \dom(E_2)$.
\end{lemma}
\begin{proof}
  By structural induction on $t$.

  The cases $t = \lambda x.y$ with $y \neq x$, $t = x u'$ and $t = u x$ are not possible because $t$ is closed.

  If $t = \lambda x. x$ then $\crumb{t} = \ess{\lambda x. \ess{x}}$ and $\fv(\lambda x. \ess{x}) = \emptyset$.

  If $t = \lambda x. t'$ then $\crumb{t} = \ess{\lambda x. \crumb{t'}}$ that is closed by~\autoref{lem:crumb-closed-prop}.

  If $t = u u'$ where $u,u'$ are not variables, then $\crumb{t} = \ess{xy} \es{x}{b'}E' \es{y}{b''}E''$; by i.h. on $u, u'$ the claim is true for $z \in \dom(\es{x}{b'}E')$ and $z \in \dom(\es{y}{b''}E'')$.
  We just need to show that $x, y$ are closed in $\es{x}{b'} E' \es{y}{b''} E''$ but this is trivial as the environment contains $\es{x}{b'}$ and $\es{y}{b''}$. \qed
\end{proof}

\begin{lemma}
  \label{lem:crumb-reachable-closed}
  Given a closed term $t$, any reachable crumble $E$ starting from $\crumb{t}$ is closed.
\end{lemma}
\begin{proof}
  We proceed by induction on the number of steps to reach $E$.

  In the base case $E$ is the initial crumble $\crumb{t}$ that by \autoref{lem:crumb-closed-prop} $E$ is closed.

  In the inductive step, $E$ is the rhs of the reduction rule i.e. $E_p \to_a E, a \in \{m_1, m_2\}$ , so the hypothesis holds for $E_p$. Now we proceed by cases on $a$.

  If $a = m_1$ then $E_p = E_1 \es{z}{xy} E_v \mone E_1 \es{z}{b'}E' E_v$ it suffices to show that $\es{z}{b'}E'$ is closed in $E_v$. 
  By i.h. $E_p$ is closed and by \autoref{lem:closed-lem-1}, if $E_v = E'_v \es{x}{\lambda x_1. \ess{b}E} E''_v$ then $\lambda x_1. \ess{b} E$ is closed in $E''_v$ and so it is closed in $E_v$; moreover, $y$ is closed in $E_v$.
  Then $\ess{b}E$ is open in $E_v$ because $x_1$ appears as a free variables; however, $(\ess{b}E) \sub{x_1}{y} = \ess{b'}E'$ is closed in $E_v$ as $y$ is closed in $E_v$.

  If $a = m_2$ then the proof is analogous to the previous case.
  \qed
\end{proof}

\crumblingharmony*
\begin{proof}~\\
  \noindent{($\Rightarrow$)} Combining \autoref{lem:norm-crumb-implies-v} with \autoref{lem:crumb-reachable-closed}.\\
	\noindent{($\Leftarrow$)} Let $E$ be a v-crumble: no reduction rule is applicable, so $E$ is normal.
	\qed
\end{proof}

\begin{lemma}
  \label{lem:crumb-t-well-named}
  For every term $t$, $\crumb{t}$ is well-named.
\end{lemma}
\begin{proof}
  By induction on $t$.

  If $t := \lambda x. y$ then $\crumb{t} = \ess{\lambda x.\ess{y}}$ that is well-named.

  If $t := \lambda x. t'$ then $\crumb{t} = \ess{\lambda x. \crumb{t'}}$ that is also well-named.

  If $t := xy$ then $\crumb{t} = \ess{xy}$ that is also well-named.

  If $t := u u'$ for $u,u'$ not variables, then $\crumb{t} = \ess{x y} \es{x}{b}E \es{y}{b'}E'$ with $\crumb{u} = \ess{b}E, \crumb{u'} = \ess{b'}E'$. 
  By i.h. on $u, u'$, $\ess{b}E$ and $\ess{b'}E'$ are well-named.
  We show that $E'' := \ess{x y} \es{x}{b}E \es{y}{b'}E'$ is well-named, i.e.
  \begin{enumerate}
   \item all the variables in $\dom(E'')$ are pairwise distinct: this holds because all variables put in
     the domain by crumbling are globally fresh
   \item $\dom(E) \cap \bv(E) = \emptyset$: this also holds for the same reason
   \item if $E'' = E_1\es{z}{b''}E_2$ then $z \not\in \fv(b'') \cup \fv(E_2)$: we proceed by cases
     on where $z$ can occur in $E''$
      \begin{itemize}
       \item Case where $z = *$ and $E_1 = \epsilon$: trivial since $*$ can only occur in the domain
       \item Case where $z \in \dom(\es{y}{b'}E')$: it holds because $\es{y}{b'}E'$ is well-named by i.h.
       \item Case where $z \in \dom(\es{x}{b}E)$: it holds because $\es{x}{b}E$ is well-named by i.h.
         and because when computing $\crumb{u}$ the variables put in the domain by crumbling are globally
          fresh also with respect to $y$ and to $u'$ and by~\autoref{rem:fv-term-crumb},
          $\fv(u') = \fv(\crumb{u'}) = \fv(\ess{b'}E')$. 
      \end{itemize}
  \end{enumerate}
  \qed
\end{proof}

\begin{lemma}
  \label{lem:term-size}
  For any term $t$, $\size{\crumb{t}} \le 2 \cdot \size{t}$
\end{lemma}
\begin{proof}
  By structural induction on $t$.

  If $t := \lambda x. y$ then $\size{\crumb{t}} = 2 = \size{t} \leq 2 \cdot \size{t}$\\~

  If $t := \lambda x. t'$ then 
  \begin{alignat*}{3}
    \size{\crumb{t}} 
    & = \size{\ess{\lambda x. \ess{\crumb{t'}}}} \\
    & = 1 + \size{\crumb{t'}} \\ 
    & \le 1 + 2 \cdot \size{t'} && \text{ by i.h. on } t' \\
    & \le 2 \cdot (1 + \size{t'}) \\
    & = 2 \cdot \size{t}
  \end{alignat*}

  If $t := u u'$ then
  \begin{alignat*}{3}
    \size{\crumb{t}}
    & = \size{\ess{xy} \es{x}{b}E \es{y}{b'}E'} \\ 
    & = 2 + \size{\crumb{u}} + \size{\crumb{u'}} \\
    & \le 2 + 2 \cdot \size{u} + 2 \cdot \size{u'} && \text{ by i.h. on } u, u' \\
    & = 2 \cdot \size{t}
  \end{alignat*}

  If $t := uy$ then 
  \begin{alignat*}{3}
    \size{\crumb{t}}
    & = \size{\ess{xy} \es{x}{b}E} \\ 
    & = 2 + \size{\crumb{u}} \\
    & \le 2 + 2 \cdot \size{u} && \text{ by i.h. on } u \\
    & \le 2 \cdot \size{t}
  \end{alignat*}

  If $t := x u'$ then 
  \begin{alignat*}{3}
    \size{\crumb{t}}
    & = \size{\ess{xy} \es{y}{b'}E'} \\ 
    & = 2 + \size{\crumb{u'}} \\
    & \le 2 + 2 \cdot \size{u'} && \text{ by i.h. on } u \\
    & \le 2 \cdot \size{t} \tag*{\qed}
  \end{alignat*}
\end{proof}

\begin{lemma}
	\label{lem:crumb-t-rvc}
  For every term $t$, if $\crumb{t} = E_1 \es{z}{xy} E_2$ then $\rback{(E_1 \es{z}{\hole})}$ is a right v-context.
\end{lemma}
\begin{proof}
	By structural induction on the shape of $t$. \\
        If $t = uu'$ where $u,u'$ are not variables, then by definition $\crumb{uu'} = \es{*}{x'y'} \es{x'}{b} E \es{y'}{b'} E'$ for $\crumb{u} = \ess{b}E$ and $\crumb{u'} = \ess{b'}E'$.
        Now, the ES of the form $\es{z}{xy}$ can appear in different segments of the crumbled term:
	\begin{itemize}
		\item As $\es{*}{x'y'}$: trivially $\rback{\es{*}{\hole}} = \hole$.
		\item As $\es{x'}{x y}$, i.e. $\es{*}{x'y'} \es{x'}{x y}$: trivially, as $\rback{(\es{*}{x'y'} \es{x'}{\hole})} = \hole y$ that is a right v-context.
		\item Anywhere in $E$, i.e. $E = E_1 \es{z}{x y} E_2$, then
      \begin{alignat*}{2}
          & \rback{(\es{*} {x'y'} \es{x'}{b}E_1 \es{z}{\hole})} \\
        = & \rback{\es{*}{x'y'}} \sub{x'}{\rback{( \es{*}{b} E_1 \es{z}{\hole} )}} \\
        = & Ry && \text{ by i.h. on } u
      \end{alignat*}
		\item As $\es{y'}{x y}$ i.e. $\es{*}{x'y'} \es{x'}{b} E \es{y'}{x y} $: in this case 
      \begin{align*}
          & \rback{( \es{*}{x'y'} \es{x'}{b} E \es{y'}{\hole} )} \\
        = & \rback{( \es{*}{x'y'} \es{x'}{b} E )} \sub{y'}{\hole} \\
        = & \rback{\es{*}{x'y'}} \sub{x'}{\rback{( \es{*}{b} E )}} \sub{y'}{\hole} \\
        = & \crback{u} \hole
      \end{align*}
		\item Anywhere in $E'$, i.e. $E' = E_2 \es{z}{xy} E_2'$, then 
      \begin{alignat*}{3}
          & \rback{(\es{*}{x'y'} \es{x'}{b} E \es{y'}{b'} E_2 \es{z}{\hole})} \\
        = & \rback{( \es{*}{x'y'} \es{x'}{b} E)} \sub{y'}{ \es{*}{b'} E_2 \es{z}{\hole} } \\
        = & x'y' \sub{x'}{\rback{( \es{*}{b} E)}} \sub{y'}{ \es{*}{b'} E_2 \es{z}{\hole} } \\
        = & \crback{u} R && \text{ by i.h. on }  u'
      \end{alignat*}
	\end{itemize}
	The cases where $t = uy$ or $t = xu'$ are simplified versions of the previous one. \\
        The cases where $t$ is a $\lambda$-abstraction are absurd because $\crumb{u}$ cannot
        be split into $E_1 \es{z}{xy} E_2$.
	\qed
\end{proof}

\crumblingprop*
\begin{proof}~
  \begin{enumerate}
    \item See \autoref{lem:crumb-t-well-named}
    \item See \autoref{lem:crumb-closed-prop}.
    \item See \autoref{lem:subterm-prop}.
    \item See \autoref{lem:term-size}.
    \item See \autoref{lem:crumb-t-rvc}. \qed
  \end{enumerate}
\end{proof}

\begin{definition}[$\sigma$-operator]
  The function $\sigma_{(\cdot)}$ turns a v-environment into the parallel substitution obtained by firing all explicit substitutions at once. It is defined by structural
  recursion as follows:
        \[ \begin{array}{ll}
          \sigma_\epsilon & = x \mapsto x \\
           \sigma_{[x\leftarrow v]E'_v} & =
		\begin{cases}
                        x \mapsto \rback{v}\;\sigma_{E'_v} \\
			y \mapsto y\;\sigma_{E'_v}
		\end{cases}
          \end{array}
	\]
\end{definition}

\begin{lemma}[Read-back to value]
	\label{lem:rback-to-value}
	For any $E_v$ that is both a crumble and a v-environment, $\rback{E_v}$ is a value.
\end{lemma}
\begin{proof}
	By structural induction on $E_v$.
	
  If $E_v = \ess{v}$ then $\rback{\es * v} = \rback v$ which is a value.
	
  If \( E_v = E'_v \es x {v} \) then 
  $\rback{E_v} = \rback{E'_v}\sub{x}{v}$; by i.h. $\rback{E'_v}$ is a value, so also $\rback{E_v} = v \sub{x}{v}$ is a value.
	\qed
\end{proof}

\begin{definition}[Disjointedness]
  We say that two environments $E, E'$ are disjoint -- using the notation $E\;\#\;E'$ -- if $\fv(E) \cap \dom(E') = \emptyset$.
\end{definition}

\begin{lemma}
	\label{lem:disjoint-env-sub}
	For any crumbled term in the form $E \es{x}{b} E'$ where $E\;\#\;E'$, then $\rback{( E \es{x}{b} E' )} = \rback{E} \sub{x}{\rback{( \es{*}{b}E' )}}$.
\end{lemma}
\begin{proof}
	Inductively on the structure of $E'$.

	If $E' = \epsilon$ then
	$ \rback{( E \es{x}{b} )} = \rback{E} \sub{x}{\rback{b}} = \rback{E} \sub{x}{\rback{\es{*}{b}}}$.
	
  If $E' = E'' \es{x'}{b'}$, then
  \begin{alignat*}{3} 
    \rback{(E \es{x}{b} E'' \es{x'}{b'} )} 
      & = \rback{(E \es{x}{b} E'' )} \sub{x'}{b'} \\
      & = \rback{E} \sub{x}{\rback{(\es{*}{b} E'')}} \sub{x'}{b'} \quad && \text{ by i.h. on } E'' \\
      & = \rback{E} \sub{x}{\rback{(\es{*}{b} E'')} \sub{x'}{b'}} \quad && \text{ since $x' \notin \fv(E)$} \\
      & = \rback{E} \sub{x}{\rback{(\es{*}{b} E'' \es{x'}{b'})}}. \tag*{\qed}
  \end{alignat*}
\end{proof}

\begin{lemma}
	\label{lem:envrback-with-sub}
	Let \( E_1 \) be any crumble and \( E_v \) be a v-environment. Then \( \rback{(E_1E_v)} = \rback{E_1} \sigma_{E_v} \).
\end{lemma}
\begin{proof}
	Inductively on the structure of $E_v$.

        If $E_v = \epsilon$ then $\sigma_{E_v} = x \mapsto x$ and trivially $\rback{E_1} = \rback{E_1}\sigma_{E_v}$.
	
  If $E_v = \es{x}{v}E'_v$ then
  \begin{alignat*}{3}
    \rback{( E_1E_v )} 
      & = \rback{( E_1\es{x}{v} )} \sigma_{E'_v} \quad && \text{ by i.h. on } E'_v \\
            & = \rback{E_1} \sub{x}{\rback{v}} \sigma_{E'_v} \\
      & = \rback{E_1} \sigma_{E_v} \quad && \text{ by definition of } \sigma \tag*{\qed}
  \end{alignat*}
\end{proof}

\begin{lemma}[Contextual Decoding]
  \label{lem:context-decoding-reachable}
	If $E_1 \es{z}{xy} E_2$ is reachable, then $\rback{(E_1 \es{z}{\hole})}$ is a right v-context.
\end{lemma}
\begin{proof}
	By induction on the reachability of the crumble. The base case is any crumbled term $\crumb{t}$ and the statement holds by~\autoref{lem:crumb-t-rvc}. \\
  For the inductive step, a crumble is reachable either by $\mone$ or by $\mtwo$.
	 
  If the crumble is reached by $\mone$, i.e. $E'_1 \es{z'}{x'y'} E'_v \mone E'_1 \es{z'}{b''} E'' E'_v$ where $E'_v(x')^\alpha = \lambda x_1. \es{*}{b}E$ then, by~\autoref{cor:body-is-a-crumble},
  $E'_v(x')$ is the crumbling of an abstraction and thus by~\autoref{lem:crumble-commute-alpha}
  $E'_v(x')^\alpha = \crumb{\lambda x_1. t} = \lambda x_1. \crumb{t}$ for some $t$.
  Therefore 
    \begin{alignat*}{3}
      & \crumb{t} \sub{x_1}{y'} \\
      = \; & \crumb{t \sub{x_1}{y'}} && \text{ by \autoref{lem:crumble-commute-free} } \\
      = \; & \ess{b''} E''
    \end{alignat*}
    If $E_1\es{z}{xy}$ is a prefix of $E'_1\es{z'}{b'}$ then the claim trivially follows by i.h. on the lhs of the multiplicative rule. Otherwise $E_1\es{z}{xy} = E'_1\es{z'}{b''}E'''\es{z}{xy}$ for $E'''\es{z}{xy}$ a prefix of $E''$. Then $\rback{E_1\es{z}{\hole}} = \rback{E'_1\es{z'}{b''}E'''\es{z}{\hole}}
    = \rback{E'_1\es{z'}{\hole}}\plug{\rback{\ess{b''}E'''\es{z}{\hole}}}$.
    By i.h. on the lhs of the multiplicative rule, $\rback{E'_1\es{z'}{\hole}}$ is a right v-context;
    by~\autoref{lem:crumb-t-rvc} on $\crumb{t \sub{x_1}{y'}}$ also $\rback{\ess{b''}E'''\es{z}{\hole}}$ is a
    right v-context. The conclusion follows by composition of right v-contexts (\autoref{lem:rvc-comp}).
		
  If the crumble is reached by $\mtwo$, i.e. $E'_1 \es{z'}{x'y'} E'_v \mtwo E'_1 \es{z'}{E'_v(y'')} E'_v$, then the claim trivially follows by i.h. since $E_1\es{z}{xy}$ must be a prefix of $E'_1$. \qed
\end{proof}

\begin{lemma}
  Given a crumble $\es{*}{b}E$ and an immediate substitution $\sub{z}{y}$ such that $z \not\in \dom(E)$, then $\rback{( \es{*}{b}E )} \sub{z}{y} = \rback{( (\es{*}{b}E) \sub{z}{y} )}$.
\end{lemma}
\begin{proof}
	By induction on $E$:
	\begin{itemize}
		\item If $E := \epsilon$ we distinguish the following cases:
		      \begin{itemize}
			      \item If $b := z_1 z_2, z_1 \ne z, z_2 \ne z$, then it is trivial as the substitution has no effect. The same principle applies if $b := \lambda z. E_1$ as the immediate substitution would have no effect.
			      \item If $b := z z_1$ then
              \begin{align*}
                \rback{\es{*}{z z_1}} \sub{z}{y}
              = & z z_1 \sub{z}{y} \\
              = & y z_1 \\
              = & \rback{\es{*}{y z_1}} \\
              = & \rback{\es{*}{z z_1} \sub{z}{y}}
              \end{align*}

			            The case where $b := z_1 z$ is analogous.
			      \item If $b := \lambda x_1. E_1$ then $\rback{\es{*}{\lambda x_1. E_1}} \sub{z}{y}
				            = (\lambda x_1. \rback{E_1}) \sub{z}{y}
				            = \lambda x_1. \rback{E_1} \sub{z}{y}$
			            that by inductive hypothesis on $E_1$ is equal to $\lambda x_1. \rback{E_1 \sub{z}{y}}
				            = \rback{\es{*}{\lambda x_1. E_1 \sub{z}{y}}}
				            = \rback{(\es{*}{\lambda x_1. E_1}) \sub{z}{y}}$.
		      \end{itemize}
		\item If $E := E' \es{z_1}{b'}$, then $z \neq z_1$ and thus
      \begin{alignat*}{2}
        & \rback{( \es{*}{b} E' \es{z_1}{b'} )} \sub{z}{y} \\
          = & \rback{(\es{*}{b} E')} \sub{z_1}{b'} \sub{z}{y} \\
          = & \rback{(\es{*}{b} E')} \sub{z}{y} \sub{z_1}{b' \sub{z}{y}} \quad && \text{since $z \neq z_1$}\\
          = & \rback{( \es{*}{b} E' \sub{z}{y} )} \sub{z_1}{b' \sub{z}{y}} \quad && \text{by i.h. } \\
          = & \rback{( \es{*}{b} E' \sub{z}{y} \es{z_1}{b' \sub{z}{y}} )} \\
          = & \rback{( (\es{*}{b} E' \es{z_1}{b'}) \sub{z}{y} )} \quad && \text{since $z \neq z_1$} \tag*{\qed}
      \end{alignat*}
	\end{itemize}
\end{proof}

\begin{lemma}
  \label{lem:multiplug}
  For every variable $z$ and $\lambda_{Plot}$-terms $t$ and $s$ such that
  $t\sub{z}\hole$ is a right v-context, $(t\sub{z}\hole)\plug{s} = t\sub{z}{s}$.
\end{lemma}
\begin{proof}
 By structural induction over $t$:
 \begin{itemize}
  \item Case $z$: $(z\sub{z}{\hole})\plug{s} = \hole \plug{s} = s = z\sub{z}{s}$
  \item Case $x$ where $x \neq z$: impossible since $x\sub{z}{\hole} = x$ that is not a right v-context
  \item Case $\lambda z.t$: impossible since $(\lambda z.t)\sub{z}{\hole} = \lambda z.t$ that is not
         a right v-context
  \item Case $\lambda x.t$ where $x \neq z$: impossible since $(\lambda x.t)\sub{z}{\hole} = \lambda x.t\sub{z}{\hole}$ that is not a right v-context
  \item Case $t_1t_2$:
    we know by hypothesis that $(t_1t_2)\sub{z}{\hole} = (t_1\sub{z}{\hole})(t_2\sub{z}{\hole})$ is
     a right v-context. We distinguish two cases:
     \begin{itemize}
       \item Case where $t_2\sub{z}{\hole}$ is a term (and thus $t_2\sub{z}{\hole}=t_2$ and
         so $z \not\in \fv(t_2)$) and $t_1\sub{z}{\hole}$ is a right v-context and thus by i.h. 
         $(t_1\sub{z}{\hole})\plug{s} = t_1\sub{z}{s}$:
         \begin{alignat*}{2}
             & ((t_1t_2)\sub{z}{\hole})\plug{s}\\
           = & ((t_1\sub{z}{\hole})(t_2\sub{z}{\hole}))\plug{s}\\
           = & ((t_1\sub{z}{\hole})\plug{s})((t_2\sub{z}{\hole})\plug{s})\\
           = & (t_1\sub{z}{s})t_2 & \text{by i.h. and $z \not\in \fv(t_2)$}\\
           = & (t_1t_2)\sub{z}{s}
         \end{alignat*}
       \item Case where $t_1\sub{z}{\hole}$ is a term (and thus $t_1\sub{z}{\hole}=t_1$ and
         so $z \not\in \fv(t_1)$) and $t_2\sub{z}{\hole}$ is a right v-context and thus by i.h. 
         $(t_2\sub{z}{\hole})\plug{s} = t_2\sub{z}{s}$:
         \begin{alignat*}{2}
             & ((t_1t_2)\sub{z}{\hole})\plug{s}\\
           = & ((t_1\sub{z}{\hole})(t_2\sub{z}{\hole}))\plug{s}\\
           = & ((t_1\sub{z}{\hole})\plug{s})((t_2\sub{z}{\hole})\plug{s})\\
           = & t_1(t_2\sub{z}{s}) & \text{by i.h. and $z \not\in \fv(t_1)$}\\
           = & (t_1t_2)\sub{z}{s} \tag*{\qed}
         \end{alignat*}
     \end{itemize}
 \end{itemize}
\end{proof}

\begin{lemma}
  \label{lem:lem4}
  For all environment context $E$ such that its read-back $\rback{E}$ is a right v-context and $E$
  begins with a crumble, then $\rback{(E \plug{b})} = \rback{E} \plug{\rback{b}}$
\end{lemma}
\begin{proof}
 We proceed by cases on the shape of $E$:
 \begin{itemize}
   \item Case $\ess{\hole}$: $$\rback{(\ess{\hole}\plug{b})} = \rback{\ess{b}} = \rback{b}  = \hole\plug{\rback{b}} = \rback{\ess{\hole}} \plug{\rback{b}}$$
   \item Case $E' \es{z}\hole$:
     \begin{alignat*}{2}
         & \rback{(E' \es{z}\hole \plug{b})}\\
       = & \rback{(E'\es{z}{b})}\\
       = & \rback{E'}\sub{z}{\rback{b}} & \text{by~\autoref{lem:multiplug}}\\
       = & (\rback{E'}\sub{z}{\hole})\plug{\rback{b}}\\
       = & \rback{(E'\es{z}{\hole})}\plug{\rback{b}} \tag*{\qed}
    \end{alignat*}
 \end{itemize}
\end{proof}

\begin{lemma}
	\label{lem:context-plug-sigma}
	Let $R$ be a right v-context and $\sigma_{E_v}$ a parallel substitution. Then $R \plug{t} \sigma_{E_v} = R \sigma_{E_v} \plug{t \sigma_{E_v}}$.
\end{lemma}
\begin{proof}
  By structural induction on $R$. 
  
  In the base case $R := \hole$ are trivial, as $R \plug{t} = t$. 

  For the inductive case, if $R := u R'$ then
  \begin{alignat*}{3}
      & R \plug{t} \sigma_{E_v} \\
    = \; & (u R' \plug{t}) \sigma_{E_v} \\
    = \; & u \sigma_{E_v} (R' \plug{t} \sigma_{E_v}) \\
    = \; & u \sigma_{E_v} (R' \sigma_{E_v} \plug{t \sigma_{E_v}}) && \text{ by i.h. on } R' \\ 
    = \; & R \sigma_{E_v} \plug{t \sigma_{E_v}}
  \end{alignat*}

  If $R := R' v$ then
  \begin{alignat*}{3}
      & R \plug{t} \sigma_{E_v} \\
    = \; & ( R' \plug{t} v) \sigma_{E_v} \\
    = \; & (R' \plug{t} \sigma_{E_v}) v \sigma_{E_v} \\
    = \; & (R' \sigma_{E_v} \plug{t \sigma_{E_v}}) v \sigma_{E_v} && \text{ by i.h. on } R' \\ 
    = \; & R \sigma_{E_v} \plug{t \sigma_{E_v}} \tag*{\qed}
  \end{alignat*}
\end{proof}

\begin{lemma}
  \label{lem:reachable-crumble-well-named}
  All reachable crumbles $E$ are well-named.
\end{lemma}
\begin{proof}
  By cases on reachability of $E$.

  If $E$ was reached by $m_1$, i.e. $E_0 = E_1 \es{z}{x y} E_v \mone E_1 \es{z}{b'}E' E_v = E$ with $E_v(x)^\alpha = \lambda x_1. \ess{b}E$ then by i.h. on $E_0, z \notin \fv(E_v)$ and so $z \notin \fv(b') \cup \fv(E') $. 
  Moreover, $E_v(x)^\alpha = \lambda x_1. \ess{b}E$ is well-named by definition of $E_v(x)^\alpha$.
  
  If $E$ was reached by $m_2$, i.e. $E_0 = E_1 \es{z}{v} E_v \mtwo E_1 \es{z}{E_v(y_1')} E_v = E$ by i.h. on $E_0, z \notin \fv(E_v)$ and so $z \notin \fv(E_v(y_1'))$.
  \qed
\end{proof}

\begin{lemma}
  \label{lem:lem3}
  For every well-named v-environment $E_v$ and any any variable $x$, $\sigma_{E_v}(x) = \rback{E_v(x)} \sigma_{E_v}$
\end{lemma}
\begin{proof}
  By structural induction on $E_v$.

  If $E_v = \epsilon$ then it is trivial as $\sigma_{E_v}$ is the identity function and $\rback{x}=x$.

  If $E_v = \es{y}{v'}E'_v$ then $\sigma_{E_v}(x) = \sigma_{E'_v}(x) = \rback{E'_v(x)} \sigma_{E'_v(x)}$ by i.h. on $E'_v$; moreover, by well-naming, $y \notin \fv(v) \cup \fv(E'_v)$, thus in particular $y \notin \fv(E'_v(x))$
  and so $\rback{E'_v(x)} \sigma_{E'_v(x)} = \rback{E'_v(x)} \sigma_{E_v(x)} = \rback{E_v(x)} \sigma_{E_v(x)}$.

  If $E_v = \es{x}{v}E'_v$ then $E_v(x) = v$; by well-naming $x \notin \fv(v)$ hence
  $\sigma_{E_v}(x) = \rback{v} \sigma_{E'_v}= \rback{v} \sigma_{E_v} = \rback{E_v(x)}\sigma_{E_v}$.
  \qed
\end{proof}

\begin{lemma}
  \label{lem:lem2}
	Given a reachable crumble $E = E_1\es{z}{xy} E_v$, its read-back is
	\[\rback{( E_1 \es{z}{xy} E_v )} = \rback{( E_1 \es{z}{\hole} )} \sigma_{E_v} \plug{ \sigma_{E_v}(x)\; \sigma_{E_v}(y)} \]
\end{lemma}
\begin{proof}
  $$
  \begin{array}{lll}
    & \rback{E} \\
  = & \rback{E_1\es{z}{xy} E_v} \\
  = & \rback{E_1\es{z}{xy}}\sigma_{E_v} & \mbox{by~\autoref{lem:envrback-with-sub}}\\
  = & \rback{E_1\es{z}{\hole}\plug{xy}}\sigma_{E_v} \\
  = & \rback{E_1\es{z}{\hole}}\sigma_{E_v}\plug{\rback{xy}\sigma_{E_v}} & \mbox{by~\autoref{lem:context-plug-sigma}}\\
    = & \rback{E_1\es{z}{\hole}}\sigma_{E_v}\plug{\sigma_{E_v}(x)\sigma_{E_v}(y)}
  \end{array}
  $$ \qed
\end{proof}

\begin{lemma}
  For any closed $\lambda_{Plot}$-term $t$, $\len{\crumb{t}} \le \size{t}$.
  \label{lem:lenlen}
\end{lemma}
\begin{proof}
  We proceed by structural induction on $t$.
  \begin{itemize}
   \item Case $x$: impossible since $t$ is closed
   \item Case $\lambda x.t'$: $\crumb{\lambda x.t'}$ is of the form $\ess{E}$ and
     $$\len{\ess{E}} = 1 \leq 1 + \len{t'} = \len{\lambda x.t}$$
   \item Case $xy$: $\crumb{xy} = \ess{xy}$ and $\len{\ess{xy}} = 1 \leq 3 = \size{xy}$
   \item Case $uy$: $\crumb{uy} = \ess{xy}\es{x}{b}E$ where $\crumb{u}=\ess{b}E$ and thus, by i.h.,
     $\len{\ess{b}E} \leq \size{u}$. We have
      $$\len{\ess{xy}\es{x}{b}E} = 1 + \len{\es{x}{b}E} \leq 1 + \size{u} \leq 2 + \size{u} = \size{uy}$$
   \item Case $xu'$: $\crumb{xu'} = \ess{xy}\es{y}{b'}E'$ where $\crumb{u'}=\ess{b'}E'$ and thus, by i.h.,
     $\len{\ess{b'}E'} \leq \size{u'}$. We have
      $$\len{\ess{xy}\es{y}{b'}E'} = 1 + \len{\es{y}{b'}E'} \leq 1 + \size{u'} \leq 2 + \size{u'} = \size{xu'}$$
   \item Case $uu'$: $\crumb{uu'} = \ess{xy}\es{x}{b}E\es{y}{b'}E'$ where $\crumb{u}=\ess{b}E$ and
         $\crumb{u'}=\ess{b'}E'$ and thus, by i.h., $\len{\ess{b}E} \leq \size{u}$ and
         $\len{\ess{b'}E'} \leq \size{u'}$. We have
         \begin{alignat*}{1}
                & \len{\ess{xy}\es{x}{b}E\es{y}{b'}E'}\\
           =    & 1 + \len{\es{x}{b}E} + \len{\es{y}{b'}E'}\\
           \leq & 1 + \size{u} + \size{u'}\\
           =    & \size{uu'} \tag*{\qed}
         \end{alignat*}
  \end{itemize}
\end{proof}

\crumblinginvariants*
\begin{proof}~
  \begin{enumerate}
    \item See \autoref{lem:reachable-crumble-well-named}.
    \item See \autoref{lem:crumb-reachable-closed}.
    \item See \autoref{lem:subterm-prop}.
    \item See \autoref{lem:context-decoding-reachable} \qed
  \end{enumerate}
  
\end{proof}

\subsection{Proofs for \autoref{sec:crumbled-calculus}}\label{app:proofs-crumbled-implementation}
\begin{lemma}[Initialization]
	\label{lem:init-crumble}
	For every closed term $t$, $\crback{t} = t$.
\end{lemma}
\begin{proof}
	Inductively on the shape of t:
	\begin{itemize}
		\item \( t := x \) is not possible since $t$ is closed by hypothesis
                \item $t := \lambda x. y$, then $\crback{(\lambda x. y)} = \rback{\ess{\lambda x. \es * y}} = \lambda x. y $
                \item $t := \lambda x. u$, then $\crback{(\lambda x. u)} = \rback{\ess{\lambda x.\crumb{u}}} = \lambda x. \crback{u}$; by i.h. $\crback u = u$, so $\lambda x. \crback{u} = \lambda x. u$
                \item \( t := u u' \) where $u,u'$ are not variables and
                  $\crumb{u}=\ess{b}E$ and $\crumb{u'}=\ess{b'}E'$
                  $$
                  \begin{array}{lll}
                      & \crback{u u'}\\
                    = & \rback{\es{*}{xy}\es{x}{b}E\es{y}{b'}{E'}}\\
                    = & \rback{\es{*}{xy}}\sub{x}{\rback{\ess{b}{E}}}\sub{y}{\rback{\ess{b'}{E'}}} & \mbox{by~\autoref{lem:disjoint-env-sub}}\\
                    = & \rback{\es{*}{xy}}\sub{x}{\crback{u}}\sub{y}{\crback{u'}}\\
                    = & xy\sub{x}{u}\sub{y}{u'} & \mbox{by i.h. on $u$ and $u'$}\\
                    = & uu'
                  \end{array}$$
                \item \( t := u x \) and \( t := x u' \) are similar but simpler than the previous case\qed
	\end{itemize}
\end{proof}

\begin{lemma}[Principal Projection]
	Let $c$ be a reachable crumble: if $c \to_a d$ then $\rback{c} \to_{\beta_v} \rback{d}$ with $a \in \{m_1, m_2\}$.
	\label{lem:principal-proj}
\end{lemma}
\begin{proof}
	By cases on $a$.
	\begin{itemize}
    \item If $a = m_1$, then $E_1 \es{z}{xy} E_v \mone E_1 \es{z}{b'}E'E_v $ with $E_v(x)^\alpha = \lambda x_1. \es{*}{b}$, then by \autoref{cor:body-is-a-crumble} $E_v(x)$ is the crumbling of an abstraction, and thus by \autoref{lem:crumble-commute-alpha} $E_v(x)^\alpha = \crumb{\lambda x_1. t} = \lambda x_1. \crumb{t}$ for some $t$. Then
        \begin{alignat*}{3}
		        & \rback{(E_1 \es{z}{xy} E_v)} \\
          = & \rback{( E_1 \es{z}{\hole} )} \sigma_{E_v} \plug{\sigma_{E_v}(x)\sigma_{E_v}(y)} && \text{ by \autoref{lem:lem2} } \\
          = & \rback{( E_1 \es{z}{\hole} )} \sigma_{E_v} \plug{(\rback{E_v(x)} \rback{E_v(y)}) \sigma_{E_v}} && \text{ by \autoref{lem:lem3} } \\
          = & \rback{( E_1 \es{z}{\hole} )} \sigma_{E_v} \plug{(\rback{(\lambda x_1.\crumb{t})} \rback{E_v(y)}) \sigma_{E_v}} \\
          = & \rback{( E_1 \es{z}{\hole} )} \sigma_{E_v} \plug{((\lambda x_1.t) \rback{E_v(y)}) \sigma_{E_v}} && \text{ by \autoref{lem:init-crumble} } \\
          \to_{\beta_v} &  \rback{( E_1 \es{z}{\hole})} \sigma_{E_v} \plug{(t \sub{x_1}{\rback{E_v(y)}}) \sigma_{E_v}} && \text{ by~\autoref{lem:context-decoding-reachable}} \\
          = & \rback{( E_1 \es{z}{\hole} )} \sigma_{E_v} \plug{(t \sub{x_1}{y} \sub{y}{\rback{E_v(y)}}) \sigma_{E_v}} && \text{ since  } y \notin \fv(E_v(y)) \cup \fv(t) \\
          = & \rback{( E_1 \es{z}{\hole} )} \sigma_{E_v} \plug{(t \sub{x_1}{y}) \sigma_{E_v}} \\
          = & \rback{( E_1 \es{z}{\hole} )} \sigma_{E_v} \plug{\crback{t\sub{x_1}{y}} \sigma_{E_v}} && \text{ by \autoref{lem:crumble-commute-free}} \\
          = & \rback{( E_1 \es{z}{\hole} )} \sigma_{E_v} \plug{\rback{(\ess{b'}E')} \sigma_{E_v}} \\
          = & \rback{( E_1 \es{z}{\hole} )} \plug{\rback{(\ess{b'}E')}} \sigma_{E_v} && \text{ by \autoref{lem:context-plug-sigma} and \autoref{lem:context-decoding-reachable}} \\
          = & \rback{( E_1 \es{z}{\hole}  \plug{\ess{b'}E')})} \sigma_{E_v} && \text{ by \autoref{lem:lem4} } \\
          = & \rback{( E_1 \es{z}{\rback{(\ess{b'}E')}})} \sigma_{E_v} \\
          = & \rback{E_1} \sub{z}{\rback{(\ess{b'}E')}} \sigma_{E_v} \\
          = & \rback{( E_1 \es{z}{b'} E' )} \sigma_{E_v} && \text{ by \autoref{lem:disjoint-env-sub} } \\
          = & \rback{( E_1 \es{z}{b'} E' E_v )} && \text{ by \autoref{lem:envrback-with-sub} }
        \end{alignat*}

      \item If $a = m_2$ then $E_v(x) = \lambda x_1. \ess{y_1}$ and , thus
        \begin{alignat*}{3}
		        & \rback{(E_1 \es{z}{xy} E_v)} \\
          = & \rback{( E_1 \es{z}{\hole} )} \sigma_{E_v} \plug{\sigma_{E_v}(x)\sigma_{E_v}(y)} && \text{ by \autoref{lem:lem2} } \\
          = & \rback{( E_1 \es{z}{\hole} )} \sigma_{E_v} \plug{(\rback{E_v(x)} \rback{E_v(y)}) \sigma_{E_v}} && \text{ by \autoref{lem:lem3} } \\
          = & \rback{( E_1 \es{z}{\hole} )} \sigma_{E_v} \plug{(\rback{(\lambda x_1. \ess{y_1})} \rback{E_v(y)}) \sigma_{E_v}} \\
          = & \rback{( E_1 \es{z}{\hole} )} \sigma_{E_v} \plug{((\lambda x_1. y_1) \rback{E_v(y)}) \sigma_{E_v}} \\
          \to_{\beta_v} &  \rback{( E_1 \es{z}{\hole})} \sigma_{E_v} \plug{(y_1 \sub{x_1}{\rback{E_v(y)}}) \sigma_{E_v}} && \text{ by~\autoref{lem:context-decoding-reachable}} \\
          = & \rback{( E_1 \es{z}{\hole} )} \sigma_{E_v} \plug{(y_1 \sub{x_1}{y} \sub{y}{\rback{E_v(y)}}) \sigma_{E_v}} && \text{ since  } y \notin \fv(E_v(y)) \cup \fv(y_1) \\
          = & \rback{( E_1 \es{z}{\hole} )} \sigma_{E_v} \plug{(y_1 \sub{x_1}{y}) \sigma_{E_v}} \\
          = & \rback{( E_1 \es{z}{\hole} )} \sigma_{E_v} \plug{y_1' \sigma_{E_v}}   \\
          = & \rback{( E_1 \es{z}{\hole} )} \sigma_{E_v} \plug{\rback{E_v(y_1')} \sigma_{E_v}} && \text{ by \autoref{lem:lem3}} \\
          = & \rback{( E_1 \es{z}{\hole} )} \plug{\rback{E_v(y_1')} } \sigma_{E_v} && \text{ by \autoref{lem:context-plug-sigma} and \autoref{lem:context-decoding-reachable}} \\
          = & \rback{( E_1 \es{z}{\hole} \plug{E_v(y_1') }} \sigma_{E_v} && \text{ by \autoref{lem:lem4}} \\
          = & \rback{( E_1 \es{z}{E_v(y_1')})} \sigma_{E_v} \\
          = & \rback{( E_1 \es{z}{E_v(y_1')} E_v)}  && \text{ by \autoref{lem:envrback-with-sub} }
          \tag*{\qed}
        \end{alignat*}
	\end{itemize}
\end{proof}

\begin{lemma}[Determinism]
	The transition $\to_{Cr}$ is deterministic
	\label{lem:cr-deterministic}
\end{lemma}
\begin{proof}
	It suffices to show that $\mone$ and $\mtwo$ are mutually exclusive, but this is trivial as the trigger for $\mone$ is that $E_v(x)^\alpha = \lambda x. \es{*}{b} E$, while the trigger of $\mtwo$ is that $E_v(x) = \lambda x. \es{*}{y}$.
	\qed
\end{proof}

\begin{lemma}[Halt]
	\label{lem:cr-halt}
	Let E be a closed crumble. If $E$ is normal, then $\rback{E}$ is normal.
\end{lemma}
\begin{proof}
  By \autoref{lem:crumble-harmony} $E$ is a v-crumble. By lemma \autoref{lem:rback-to-value} $\rback{E}$ is a value. By harmony $\rback{E}$ is normal.
	\qed
\end{proof}

\crumblingimplementation*
\begin{proof}
  Since we don't have any overhead transitions we only need to prove the following properties:
	\begin{enumerate}
    \item Initialization: See \autoref{lem:init-crumble}.
		\item Principal Projection: See \autoref{lem:principal-proj}.
    \item Determinism: See \autoref{lem:cr-deterministic}.
		\item Halt: See \autoref{lem:cr-halt}.  \qed
	\end{enumerate}
\end{proof}

\section{Proofs for \autoref{sec:rcam}}

\subsection{Proofs for \autoref{sec:rcam-machine}}
\invariantsrcam*
\begin{proof}
  Freshness and closure are inherited by the environments that compose the state.

  For the rightmost invariant, we only move the pointer to the next entry of the environment with the $\sear$ transition, i.e. we only move the pointer when the entry already contains a value.
  \qed
\end{proof}

\harmonyrcam*
\begin{proof}
  \noindent{($\Rightarrow$)} Suppose $S = (E\es{z}{b}, E_v, \hist)$: then if $b = xy$ a rule between $\mmoner, \mmtwor$ would be triggered since $S$ is reachable and thus, by~\autoref{lem:harmony-rcam}, $x,y \in \dom(E_v)$ by the invariant Closure and they are bound to values by the invariant Rightmost; otherwise, if $b = v$ the $\sear$ transition would be triggered. In both cases we see a contradiciton. Therefore the active environment must be empty.

	\noindent{($\Leftarrow$)} Trivially, no $\rcr$ transition cannot be triggered as the portion of the environment to evaluate is empty.
  \qed
\end{proof}

\rcamimplementation*
\begin{proof} 
  It is easy to notice that all the properties excluding the ones related to the overhead transitions are in a one-to-one relation with the ones presented in \autoref{lem:crumb-invariants}.
  
  For the overhead transparency, $(E_1\es{z}{v}, E_v, \hist) \sear (E_1, \es{z}{v} E_v, \hist)$ and
  $\rbackm{(E_1\es{z}{v}, E_v, \hist)} = E_1\es{z}{v}E_v = \rbackm{(E_1, \es{z}{v} E_v, \hist)}$ \\

  For the overhead termination, we will show in~\autoref{lem:count-sea} that the number of $\sear$ transitions is bi-linear in the number of principal transitions and the maximal length of an environment in the initial crumble, so the $\sear$ transition terminates.
  \qed
\end{proof}

\subsection{Proofs for \autoref{sec:rcam-reversibility}} \label{app:proofs-rcam-rev}
\backwarddet*
\begin{proof}
  By cases on the top of the history stack.

  If $\hist = \langle \rangle : \hist'$ then only $\seal$ can be triggered.

  If $\hist = \langle x, y \rangle : \hist'$ then we trigger either $\mmonel$ or $\mmtwol$ depending on the content of $E_v(x)$.
  \qed
\end{proof}

\welfoundedness*
\begin{proof}
  Trivially follows from the fact that each $\rcr$ transition adds an element on the top of the stack
  and thus the stack of $S_0$ would have an infinite size, which is not possible. \qed
\end{proof}

\fwdonly*
\begin{proof}
  By induction on $\size{\rho}_b$.

  If $\size{\rho}_b = 0$ then trivially $\rho'$ is $\rho$ as it doesn't contain any backward transition.

  If $\size{\rho}_b = n + 1$ then $\rho : S_0 \rcr^m S_1 \rcrinv S_2 \rcrfull^* S'$.
  The case $m=0$ is not possible since $S_1$ would be an initial state, which has an empty stack,
  and therefore $S_1 \rcrinv S_2$ would not be possible.
  Therefore
  $\rho : S_0 \rcr^{m-1} S'_2 \rcr S_1 \rcrinv S_2 \rcrfull^* S'$ and, by~\autoref{lem:id-rev},
  $S'_2 = S_2$ and thus $\rho'' : S_0 \rcr^{m-1} S_2 \rcrfull^* S'$ is obtained by $\rho$ by
  cancelling out the transition $S'_2 \rcr S_1 \rcrinv S_2$.
  Therefore $\size{\rho''}_b < \size{\rho}_b$ and $\rho'' : S_0 \rcrfull^* S'$.
  The proof follows immediately by i.h. on $\rho''$.
  \qed
\end{proof}

\harmonyrev*
\begin{proof} 
	\noindent{($\Leftarrow$)} Trivial, since the $\hist$ is empty. \\
  \noindent{($\Rightarrow$)}
   Since $S$ is reachable, there exists a $\rho : S_0 \rcrfull^* S$ where $S_0$ is initial.
   By~\autoref{lem:rev-fwd-only} there also exists a $\rho' : S_0 \rcr^m S$.
   If $m = 0$ then $S = S_0$ and $S$ is initial.
   Otherwise $\rho' : S_0 \rcr^{m-1} S' \rcr S$ and $S \rcrinv S'$ by~\autoref{lem:invertstate},
   contradicting the hypothesis that $S$ is $\rcrinv$-normal. \qed
\end{proof}

\subsection{Proofs for \autoref{sec:rcam-complexity}}\label{app:proofs-rcam-complexity}
\lenrcam*
\begin{proof}
  By induction on the length of $\rho$. 

  For the base case, $S = \initm{E} = (E, \epsilon, \epsilon)$, $\size{\rho}_p = \size{\rho}_{sea}=0$ and so trivially $\len{E} \le \len{E} + 0 \cdot L(E) - 0$.

  For the inductive case, the search transitions reduce by one the length of the unevaluated environment, while the multiplicative transitions increases the measure of the unevaluated part by a number bound by $L(E)$.
  \qed
\end{proof}

\numsea*
\begin{proof}
  By \autoref{lem:lenrcam} for $E = \crumb{t}$ and $E_1 = \epsilon$ since $\rho$ is normalizing,
  $0 \leq \len{\crumb{t}} + \size{\rho}_p \cdot L(\crumb{t}) - \size{\rho}_{sea}$.
  By \autoref{lem:lenlen}, $\len{\crumb{t}} \leq \size{t}$ and also $L(\crumb{t}) \leq \size{t}$.
  Therefore
   \begin{alignat*}{1}
          & \size{\rho}_{sea}\\
     \leq & \len{\crumb{t}} + \size{\rho}_p \cdot L(\crumb{t})\\
     \leq & \size{t} + \size{\rho}_p \cdot \size{t}\\
        = & (\size{\rho}_p + 1)\cdot\size{t} \tag*{\qed}
   \end{alignat*}
\end{proof}

\end{document}